\documentclass{article}

\usepackage[ruled,vlined]{algorithm2e}
\usepackage{subfig}
\usepackage[export]{adjustbox}
\usepackage{thmtools, thm-restate}
\usepackage[a-1b]{pdfx}
\usepackage{authblk}
\usepackage{hyperref}
\usepackage{amsmath}
\usepackage{amssymb}
\usepackage{amsthm}

\newcommand{\A}{{\scalebox{1}{\val{sep}}}}	
\newcommand{\B}{{\scalebox{1}{\val{autumn}}}}	
\newcommand{\auno}{{\scalebox{1}{\val{sep10}}}}

\newcommand{\bdue}{{\scalebox{1}{\val{oct10}}}}	
	
\newcommand{\abdue}{{\scalebox{1}{\val{sep30}}}}
\newcommand{\summer}{{\scalebox{1}{\val{summer}}}}
\newcommand{\jultwentyone}{{\scalebox{1}{\val{jul21}}}}
\newcommand{\juntwentyfour}{{\scalebox{1}{\val{jun24}}}}
\newcommand{\apr}{{\scalebox{1}{\val{apr}}}}
\newcommand{\aprten}{{\scalebox{1}{\val{apr10}}}}
\newcommand{\aprfifteen}{{\scalebox{1}{\val{apr15}}}}
\newcommand{\julten}{{\scalebox{1}{\val{jul10}}}}
\newcommand{\spring}{{\scalebox{1}{\val{spring}}}}
\newcommand{\jun}{{\scalebox{1}{\val{jun}}}}
\newcommand{\may}{{\scalebox{1}{\val{may}}}}
\newcommand{\mayseven}{{\scalebox{1}{\val{may7}}}}
\newcommand{\jul}{{\scalebox{1}{\val{jul}}}}

\newtheorem{ex}{Example}
\newtheorem{definition}{Definition}
\newtheorem{lemma}{Lemma}

\newlength{\qedlengte}
\newcommand{\qedbox}{\rule{\qedlengte}{\qedlengte}}

\newenvironment{example}{\begin{ex}\rm}{\hfill\qedbox\end{ex}}

\newcommand{\seqSet}{{\Sigma}}                      

\newcommand{\po}{\leq}                      
                    
\newcommand{\pot}{\po_t}

\newcommand{\pom}{\po_V}

\newcommand{\pov}{{\po_V}}                    
\newcommand{\povi}{{\po_{V_i}}}

\newcommand{\domain}{\mathcal{T}}

\newcommand{\Top}{\textsf{\upshape{T}}}
\newcommand{\Sop}{\textsf{\upshape{S}}}

\newcommand{\Xop}{\textsf{\upshape{X}}}
\newcommand{\Yop}{\textsf{\upshape{Y}}}

\newcommand{\Emptyseq}{\varepsilon}

\newcommand{\pref}{\succ}                       	
\newcommand{\wpref}{\succeq}                       	
\newcommand{\equp}{\approx}

\newcommand{\wpreft}{\wpref_{\Top}}               
\newcommand{\wprefs}{\wpref_{\Sop}}               
\newcommand{\wprefO}{\wpref_{\Emptyseq}}          

\newcommand{\wprefx}{\wpref_{\Xop}}

\newcommand{\wprefxt}{\wpref_{\Xop\Top}}
\newcommand{\prefxt}{\pref_{\Xop\Top}}

\newcommand{\impl}{\rightarrow}

\newcommand{\Best}{\beta}

\def\codeif{\mbox{\upshape\textbf{if}}}
\def\codethen{\mbox{\upshape\textbf{then}}}

\def\codeforeach{\mbox{\upshape\textbf{for each}}}

\def\codetrue{\mbox{\upshape\textbf{true}}}
\def\codefalse{\mbox{\upshape\textbf{false}}}
\def\codereturn{\mbox{\upshape\textbf{return}}}

\def\coderepeat{\mbox{\upshape\textbf{repeat}}}
\def\codeuntil{\mbox{\upshape\textbf{until}}}

\newcommand{\attr}[1]{\textsf{\small #1}}
\newcommand{\val}[1]{\textsf{\small #1}}
\newcommand{\lev}[1]{\texttt{\small #1}}

\definecolor{davidecolor}{rgb}{0.13, 0.55, 0.13}

\newcommand{\comment}[1]{}

\newcommand{\DiffBest}{\textsc{DiffBest}}

\newcommand{\Oh}{\mathcal{O}}

\begin{document}

\title{Preference Queries over Taxonomic Domains}

\author{Paolo Ciaccia$^{1}$ \and Davide Martinenghi$^2$ \and Riccardo
  Torlone$^3$}

\affil{$^1$University of Bologna, Italy,\quad \texttt{paolo.ciaccia{@}unibo.it}\\
$^2$Politecnico di Milano, Italy,\quad \texttt{davide.martinenghi{@}polimi.it}\\
$^3$Universit\`a Roma Tre, Italy,\quad \texttt{torlone{@}dia.uniroma3.it}\\
}
 
\date{}

\maketitle
\pagenumbering{arabic}
\pagestyle{plain}
\setcounter{page}{1}

\begin{abstract}
When composing multiple preferences characterizing the most suitable results for a user, several issues may arise.
Indeed, preferences can be partially contradictory, suffer from a mismatch with the level of detail of the actual data, and even lack natural properties such as transitivity.
In this paper we formally investigate the problem of retrieving the best results complying with multiple preferences expressed in a logic-based language.
Data are stored in relational tables with taxonomic domains, which allow the specification of preferences also over values that are more generic than those in the database.
In this framework, we introduce two operators that rewrite preferences for enforcing the important properties of transitivity, which guarantees soundness of the result, and specificity, which solves all conflicts among preferences.
Although, as we show, these two properties cannot be fully achieved together, we use our operators to identify the only two alternatives that ensure transitivity and minimize the residual conflicts. Building on this finding, we devise a technique, based on an original heuristics, for selecting the best results according to the two possible alternatives. We finally show, with a number of experiments over both synthetic and real-world datasets, the effectiveness and practical feasibility of the overall approach.
\end{abstract}

\section{Introduction}\label{sec:introduction}

Preferences strongly influence decision making and, for this reason, their collection and exploitation are considered building blocks of 
content-based filtering techniques \cite{Ch03,DBLP:journals/tods/StefanidisKP11,RicciRS11}.
A key issue in this context is the mismatch that usually lies between preferences and data, which often makes it hard to recommend items to customers~\cite{LiBe01}.
Indeed, whether they are collected by tracing the actions of the users or directly elicited from them, preferences are typically expressed in generic terms (e.g., I prefer pasta to beef), whereas available data is more specific (the menu might contain lasagne and hamburger).
The problem of automatically suggesting the best solutions becomes even more involved when several preferences at different levels of granularity and possibly conflicting with each other are specified, as shown in the following example that will be used throughout the rest of the paper.  

\begin{example}\label{ex:concerts}
We would like to select some bottles of wine from the list in Figure~\ref{fig:concerts} available in an e-commerce store.
We prefer white wines to red ones, yet we prefer Amarone (a famous red wine) to white wine. For the producer, we prefer Tuscan wineries located in the province of Siena to those in the Piedmont province of Asti. Moreover, if the winery lies in the Langhe wine region (which spans different provinces, partially including, among others, Asti and Cuneo) we prefer an aged wine (i.e., produced before 2017) to a more recent one. Finally, we would like to have suggestions only for the ``best'' possible alternatives. 
\end{example}

\begin{figure}[h]%
\begin{center}\rm%
\begin{tabular}{|ccc|c}
\multicolumn{3}{c}{\sf Wines} \\
\cline{1-3}
\textsf{Wine} &\textsf{Winery} &\textsf{Year}\\
\cline{1-3}
 Arneis & Correggia & 2019 &\ $a$\\
 Amarone & Masi & 2014 &\ $b$\\
 Amarone & Bertani & 2013 &\ $c$\\
 Canaiolo & Montenidoli & 2015 &\ $d$\\
 Barolo & Laficaia &  2014 &\ $e$\\
 Arneis & Ceretto & 2019 &\ $f$\\
\cline{1-3}
\end{tabular}
\end{center}
\caption{A list of wines}\label{fig:concerts}
\end{figure}
We first observe that further information is needed in this example to identify the solutions that better fit all the mentioned preferences. For instance, we need to know the province and the wine region in which all the wineries are located. In addition, the example shows that there are two important issues that need to be addressed in such scenarios. First, conflicts can occur when preferences are defined at different levels of detail. Indeed, the preference for Amarone, which is a red wine, is in contrast with the more generic preference for white wines. Second, further preferences can be naturally derived from those that are stated explicitly. For instance, from the preference for wines from Siena to those from Asti and the preference for aged wines when they are from the Langhe region, we can also derive, by transitivity, a preference for wines from Siena to young wines from Langhe. 

In this paper we address the problem of finding the best data stored in a repository in a very general scenario in which, as in the above example: (i) preferences may not match the level of detail of the available data, (ii) there may be conflicts between different preferences, and (iii) known preferences can imply others. Specifically, unlike previous approaches that have only tackled the problem of mapping preferences to data (see, e.g., \cite{DBLP:conf/ijcai/LukasiewiczMS13}), we formally investigate the two main principles that need to be taken into account in this context: \emph{specificity} and \emph{transitivity}. Specificity is a fundamental tool for resolving conflicts between preferences by giving precedence to the most specific ones, as it is natural in practical applications. For instance, in our example, the specific preference for Amarone over white wines counts more than the generic preference for white wines over red ones.
The specificity principle is indeed a pillar of non-monotonic reasoning, where a conclusion derived from a more specific antecedent overrides a conflicting inference based on a less specific antecedent~\cite{Horty94}. 
On the other hand, transitivity, besides being a natural property, is important also from a practical point of view, since non-transitive preferences might induce cycles, a fact that could make it impossible to identify the best solutions~\cite{Ch03}. 

To tackle the problem of dealing with non-monotonic preferences, we rely on a natural extension of the relational model in which we just assume that  \emph{taxonomies}, represented by partial orders on values, are defined on some attribute domains~\cite{MT:VLDBJ2014}. Thus, for instance, in a geographical domain we can establish that the value \texttt{Italy} is more generic than the value \texttt{Rome}, since the former precedes the latter in the partial order. We then call \emph{t-relations} (i.e., relations over taxonomies) standard relations involving attributes over these taxonomic domains. 

We express preferences in this model in a declarative way, by means of first-order \emph{preference formulas} specifying the conditions under which, in a t-relation, a tuple $t_1$ is preferable to a tuple $t_2$. By taking advantage of the taxonomies defined over the domains, in a preference formula we can refer to values that are more generic than those occurring explicitly in a t-relation (e.g., the fact that we prefer \texttt{white} to \texttt{red} wines, as in Example~\ref{ex:concerts}). 
When evaluated over a t-relation $r$, a preference formula returns a \emph{preference relation} that includes all the pairs of tuples $(t_1,t_2)$ in $r$ such that $t_1$ is preferable to $t_2$. 
Since the input preference formula may not induce a preference relation enjoying both transitivity and specificity, such a formula then needs to be suitably rewritten.
Eventually, the rewritten formula is used to select the best tuples in $r$ by means of the \emph{Best} operator, which filters out all the tuples that are strictly worse than some other tuple~\cite{DBLP:journals/jacm/CiacciaMT20}.
How this rewriting has to be performed is thus the main focus of this paper.

\noindent
\textbf{Problem.}
\emph{To study, from both a theoretical and a practical point of view, to which extent the properties of transitivity and specificity can be obtained by suitable rewritings of the initial preference formula.}

We tackle the problem by introducing and formally investigating
two operators that rewrite a preference formula: $\Top$ to enforce transitivity and $\Sop$ to remove all conflicts between more generic and more specific preferences, thus attaining specificity.
In order to try to guarantee both properties, one thus needs to use both operators.
The first natural question that arises is whether the order in which they are applied is immaterial.
Unfortunately, it turns out that these two operators do not commute.
More so, even their repeated application can produce different results, inducing incomparable preference relations.
This motivates us to explore
the (infinite) space of possible \emph{sequences} of such operators.
Based on this analysis, we
 prove that it is indeed \emph{impossible} to always guarantee at the same time transitivity of the obtained preference relation and a complete absence of conflicts therein,
no matter the order in which $\Top$ and $\Sop$ are considered and how many times they are applied.
Intuitively, the removal of conflicts may compromise transitivity, whereas enforcement of transitivity may (re-)introduce conflicts.
We also show that this impossibility result would persist even if one considered a more fine-grained $\Sop$ operator that removes conflicts one by one (instead of all at a time).
In spite of this intrinsic limitation, we formally show that: (i)~the set of all possible sequences of operators can be reduced to a finite (and small) set, and
(ii)~there are only two sequences, which we call \emph{minimal-transitive}, that guarantee transitivity and, at the same time, \emph{minimize} residual conflicts between preferences.
We also show that the application of the Best operator using the rewritten formulas obtained through the two minimal-transitive sequences can lead to very different results. 
However, in common practical cases, experimental evidence shows that one of the two sequences typically resolves more conflicts, thus returning a more refined set of best tuples.

In order to observe and assess the actual behavior of sequences of operators, we developed an engine for implementing our approach, which rewrites an input preference formula and evaluates it over t-relations.
We conducted a number of experiments over both synthetic and real-world data and taxonomies in scenarios of different complexities, showing that:
(i) the overhead incurred by the rewriting process is low for the considered sequences; (ii) the computation of the best results largely benefits from the minimization of conflicts between preferences, both in terms of execution time and cardinality of results; (iii) the adoption of an original heuristic sorting criterion based on taxonomic knowledge greatly reduces execution times.

In sum, the contributions of this paper are the following:
\begin{itemize}
\item a general framework that is able to express, in a logic-based language, preferences over relations with taxonomic domains, as illustrated in Section~\ref{sec:preliminaries};
\item two operators, presented in Section~\ref{sec:propagation}, that rewrite, within this framework, the input preferences so as to enforce the important properties of transitivity, which is required for the correctness of the result, and specificity, which solves possible conflicts among preferences;
\item the formal investigation, illustrated in Section~\ref{sec:semantics}, of the combined and repeated application of these operators to an initial set of preferences;
\item a technique based on an original heuristics, presented in Section~\ref{sec:algorithm}, for selecting the best results associated with given sequences of operators, and the characterization of their differences;
\item the experimentation of the overall approach over both synthetic and real-world data, showing its effectiveness and practical feasibility, as illustrated in Section~\ref{sec:experiments}.
\end{itemize}
Related works are reported in Section~\ref{sec:related} whereas some conclusions are sketched in  Section~\ref{sec:conclusions}.

This paper is an extended version of~\cite{DBLP:journals/pvldb/CiacciaMT21}, with formal proofs available in the appendix.

\section{Preliminaries}\label{sec:preliminaries}

In this section, we introduce our data model, originating from~\cite{MT:VLDBJ2014}, and a logic-based preference model, inspired by~\cite{Ch03}.

We remind that a \emph{partial order} $\po$ on a domain $V$ is a subset of $V\times V$, whose elements are denoted by $v_1\po v_2$, that is: 1) reflexive ($v\po v$ for all $v\in V$), 2) antisymmetric (if $v_1\po v_2$ and $v_2\po v_1$ then
$v_1=v_2$), and 3) transitive (if $v_1\po v_2$ and $v_2\po v_3$ then $v_1\po v_3$). A set with a partial order is called a \emph{poset}.

\subsection{Data Model}\label{sec:datamodel}

We consider a simple extension of the relational model in which the values of an attribute can be arranged in a hierarchical \emph{taxonomy}.

\begin{definition}[Taxonomy]\label{def:hdomain}
A \emph{taxonomy} is a poset $T=(V,\pov)$, where $V$ is a set of \emph{values} and $\pov$ is a partial order on $V$.
\end{definition}

\begin{example} \label{ex:levelmapping}
A taxonomy relevant to our working example represents production sites at different levels of granularity. Considering Example~\ref{ex:concerts}, this taxonomy, $T_p$, shown in Figure~\ref{fig:geo-tax}, includes values representing wineries (as minimal elements of the poset) as well as values representing provinces, wine regions, and regions of Italy. For instance, we can have values like $\val{Laficaia}$ (a winery), $\val{Cuneo}$ (a province), $\val{Langhe}$ (a wine region) and $\val{Piedmont}$ (a region of Italy), with $\val{Laficaia}\pom \val{Cuneo}$, $\val{Laficaia}\pom\val{Langhe}$, $\val{Laficaia}\pom\val{Piedmont}$, $\val{Cuneo}\pom\val{Piedmont}$, and $\val{Langhe}\pom\val{Piedmont}$.
Additionally, Figure~\ref{fig:vini-colori} shows a simple taxonomy $T_w$ for wines, which associates each wine with a corresponding color.
Finally, we assume a taxonomy $T_y$ mapping production years before 2017 to $\val{aged}$ and the other years to $\val{young}$.
\end{example}

\begin{figure}%
\centering
\subfloat[][{A taxonomy $T_p$ for production sites.}]
{\includegraphics[width=0.9\columnwidth]{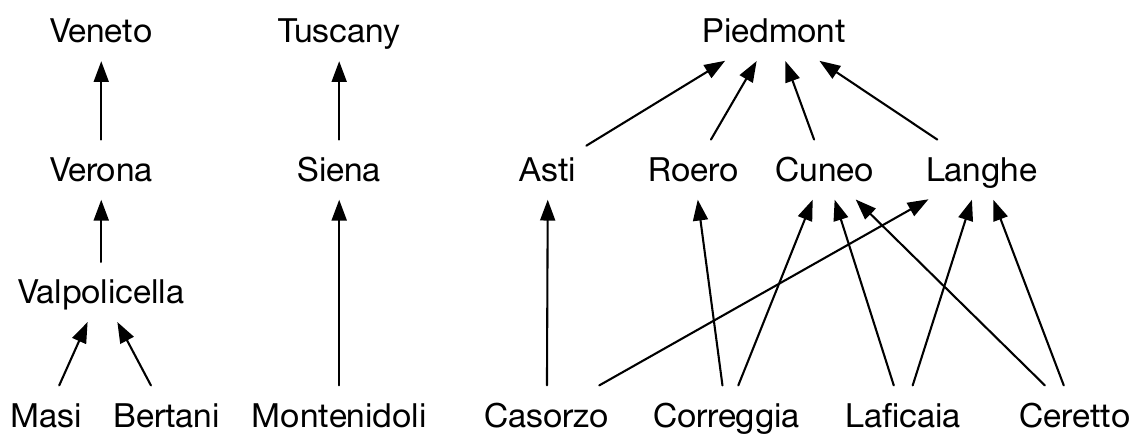}\label{fig:geo-tax}}%
\quad
\subfloat[][{A taxonomy $T_w$ for wines.}]
{\includegraphics[width=0.57\columnwidth]{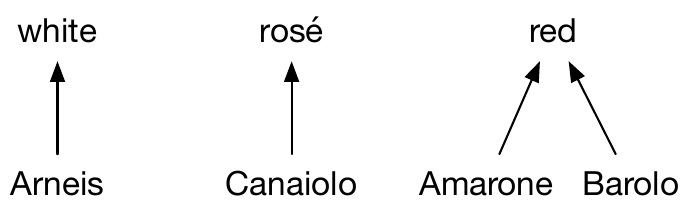}\label{fig:vini-colori}}%
\\[-2ex]
\caption{Taxonomies for the running example.}\label{fig:taxonomies-running}%
\end{figure}

A \emph{t-relation} is a standard relation of the relational model defined over a collection of taxonomies.

\begin{definition}[t-relation, t-schema, t-tuple]\label{def:trelation}
A \emph{t-schema} is a set 
$S=\{A_1:T_1,\ldots,A_d:T_d\}$, where each $A_i$ is a distinct \emph{attribute name} and each
 $T_i=(V_i,\povi)$ is a taxonomy.
A \emph{t-relation} over $S$
is
 a set of tuples over $S$ (\emph{t-tuples}) mapping each $A_i$ to a value in $V_i$. We denote by $t[A_i]$ the restriction of a t-tuple $t$ to the attribute $A_i$.
\end{definition}
For the sake of simplicity, in the following we will not make any distinction between the name of an attribute of a t-relation and that of the corresponding taxonomy, when no ambiguities can arise.
We observe that our model also accommodates ``standard'' attributes, in which the domain $V$ is a set of flat values (i.e., $\pov$ is empty).

\begin{example} \label{ex:second}
A catalog of Italian wines can be represented by the t-schema
$S=\{\lev{Wine}:T_w,\lev{Winery}:T_p, \lev{Year}:T_y\}$.
A possible t-relation over $S$ is shown in Figure~\ref{fig:concerts}.
Then we have ${b}[\lev{Year}]=\val{2014}$ and ${e}[\lev{Wine}]=\val{Barolo}$.
\end{example}

\comment{
The partial order defined on taxonomies naturally induces a partial order $\pot$ on t-tuples as follows.
\begin{definition}[Partial order on t-tuples]\label{def:pot}
Given two t-tuples $t_1$ and $t_2$ over $S_1$ and $S_2$, respectively,
with $S_2\subseteq S_1$,
we have that
$t_1\pot t_2$ if
$t_1[A_i]\povi t_2[A_i]$
for each $A_i:T_i \in S_2$.
\end{definition}
Note that, when $S_1=S_2$, the above definition corresponds to the well-known notion of product order.

\begin{example} \label{ex:posets}
\sloppy
Consider the t-schemas
$S_1=\{\lev{Artist},\lev{Time}, \lev{Venue}\}$ and
$S_2=\{\lev{Artist},$ $\lev{Time}\}$, and the t-tuples
$t_1 = \langle \val{Eddie Vedder},\val{16/04/2019},\val{Verona Arena}\rangle$ 
and $t_2 = \langle \val{rock},\val{spring}\rangle$ over $S_1$ and $S_2$, respectively.
Then, $t_1 \pot t_2$. 
\hfill$\Box$\end{example} 

}

\subsection{Preference Model }\label{sec:preferencemodel}

Given a set of attribute-taxonomy pairs $A_1:T_1,\ldots,A_d:T_d$, in which $A_1,\ldots,A_d$ are all distinct, let $\domain$ denote the set of all possible t-tuples over any t-schema that can be defined using such pairs.

\begin{definition}[Preference relation]\label{def:preference-relation}
A \emph{preference relation} over the t-tuples in $\domain$ is a
relation $\wpref$ on $\domain\times\domain$.
Given two t-tuples $t_1$ and $t_2$ in $\domain$, if $t_1\wpref t_2$ then $t_1$ is \emph{(weakly) preferable} to $t_2$, also written as $(t_1,t_2)\in\  \wpref$. If  $t_1\wpref t_2$ but $t_2\not\wpref t_1$, then $t_1$ is \emph{strictly preferable} to $t_2$, denoted by $t_1\pref t_2$.
\end{definition}

\begin{definition}[Incomparability and Indifference]\label{def:incomparability-indifference}
Given a preference relation on $\domain$ and a pair of t-tuples $t_1$ and $t_2$ in $\domain$, if neither $t_1\wpref t_2$ nor $t_2\wpref t_1$, 
then $t_1$ and $t_2$ are \emph{incomparable}.
When both $t_1\wpref t_2$ and $t_2\wpref t_1$ hold, $t_1$ and $t_2$ are \emph{indifferent}, denoted by $t_1 \equp t_2$.
\end{definition}

Notice that if $\wpref$ is transitive, then $\equp$ is an equivalence relation (up to reflexivity) and $\pref$ is a strict partial order (i.e., transitive and irreflexive). These properties do not hold, in the general case, when $\wpref$ is not transitive.

The transitivity of $\wpref$ implies that all the t-tuples involved in a cycle are indifferent to each other, thus the cycle vanishes when strict preferences are considered.
\begin{example}\label{ex:cycle}
Let us consider the t-relation in Figure \ref{fig:concerts} and assume that we have the cycle of preferences: $a \wpref b, b \wpref c$, and $c \wpref a$. If $\wpref$ is transitive then we also have  $a \wpref c$ (from $a \wpref b$ and $b \wpref c$), $b \wpref a$ (from $b \wpref c$ and $c \wpref a$) and $c \wpref b$ (from $c \wpref a$ and $a \wpref b$). Then, since $a \equp b$, $b \equp c$, and $a \equp c$, no cycle is present in $\pref$.
\end{example}

Given a set of t-tuples $r \subseteq \domain$, the ``best'' t-tuples in $r$ according to the preference relation $\wpref$ can be selected by means of the \emph{Best} operator $\Best$~\cite{DBLP:journals/jacm/CiacciaMT20}, which returns the t-tuples $t_1$ of $r$ such that there is no other t-tuple $t_2$ in $r$ that is \emph{strictly} preferable to $t_1$.

\begin{definition}[Best operator]\label{def:best}
Given a t-relation $r$ and a preference relation $\wpref$ on the t-tuples in $r$, the best operator $\Best$ is defined as follows:
$
\Best_\pref(r) 
= \{ t_1\in r \mid \nexists t_2\in r, t_2 \pref t_1\}.
$
\end{definition}
When $\pref$ is a strict partial order, $\Best_\pref(r)$ is not empty for any non-empty t-relation $r$. 
We remind that, if $\pref_1$ and $\pref_2$ are such that $\pref_1\subseteq \pref_2$ then $\Best_{\pref_2}(r)\subseteq\Best_{\pref_1}(r)$ holds for all $r$~\cite{Ch03}.

\begin{example}\label{ex:best}
Let us consider the t-relation in Figure \ref{fig:concerts} and assume that: $b \wpref a, a \wpref f, b \wpref f, b \wpref d, c \wpref e, e \wpref c$. It follows that: $b \pref a, a \pref f, b \pref f, b \pref d$ (since the opposite does not hold for those four preferences), 
but $c\equp e$ (since both $c \wpref e$ and $e \wpref c$). Then, we have $\Best_\pref(r)=\{b,c,e\}$.
\end{example} 

For expressing preferences we consider a logic-based language, in which $t_1 \wpref t_2$ iff they satisfy the first-order \emph{preference formula} $F(t_1,t_2)$:
$
t_1 \wpref t_2 \Leftrightarrow F(t_1,t_2).
$
Thus, when considering strict preferences we have:
\begin{equation}\label{eq:strict-pref}
t_1 \pref t_2 \Leftrightarrow F(t_1,t_2) \wedge \neg F(t_2,t_1).
\end{equation}
As in~\cite{Ch03}, we only consider \emph{intrinsic preference formulas} (ipf's), i.e., first-order formulas in which only built-in predicates are present and quantifiers are omitted, as in Datalog. 
Predicates have either the form
$(x[A_i]\; \povi\; v)$ or $(x[A_i]\; \not\leq_{V_i}\;  v)$, where
$A_i$ is an attribute defined over taxonomy $T_i=(V_i,\povi)$,
$x$ is a t-tuple variable over t-schemas including $A_i$,
and $v$ is a value in $V_i$. The predicate $(x[A_i]\; \povi\; v)$ (resp.\ $(x[A_i]\; \not\leq_{V_i}\;  v)$) holds for a t-tuple $t$ if $(t[A_i]\; \povi\; v)$ (resp.\ $(t[A_i]\; \not\leq_{V_i}\;  v)$) holds.
For convenience, we get rid of $\lnot$ as needed by transforming $\povi$ into $\not\leq_{V_i}$ and vice versa.

For the sake of generality, we consider that formula $F$ consists of a set of \emph{preference statements}, where each statement $P_i$ is in Disjunctive Normal Form (DNF), each disjunct of $P_i$ being termed a \emph{preference clause}, $C_{i,j}$:
\[
P_i(x,y) = \bigvee_{j=1}^{m_i} C_{i,j}(x,y)
\]
and where each clause $C_{i,j}$ is a conjunction of predicates.
We assume that each clause $C_{i,j}$ is non-contradictory, i.e., $\exists t_1,t_2\in\domain$ such that $C_{i,j}(t_1,t_2)$ is true. When a statement consists of a single clause we use the two terms ``clause'' and ``statement'' interchangeably.

A formula $F$ is a disjunction of $n \geq 1$ preference statements:
\[
F(x,y) = \bigvee_{i=1}^n P_i(x,y).
\]

\begin{example}\label{ex:pref-concerts}
The preferences informally stated in Example \ref{ex:concerts} can be expressed by the
 formula
\[F(x,y) = P_1(x,y) \vee P_2(x,y) \vee P_3(x,y) \vee P_4(x,y)\] 
where the 4 preference statements, in which we use $\leq$ in place of $\povi$ to improve readability, are: 
\[
\begin{array}{rl}
P_1(x,y) =& (x[\attr{Wine}] \leq\val{white}) \wedge 
				(y[\attr{Wine}]\leq\val{red}) \\[1mm]
P_2(x,y) =& (x[\attr{Wine}]\leq\val{Amarone}) \wedge 
				(y[\attr{Wine}]\leq\val{white}) \\[1mm]
P_3(x,y) =& (x[\attr{Winery}]\leq\val{Siena}) \wedge 
				(y[\attr{Winery}]\leq\val{Asti}) \\[1mm]
P_4(x,y) =& (x[\attr{Winery}]\leq\val{Langhe}) \wedge 
				(x[\attr{Year}]\leq\val{aged})\wedge \\
			  &	(y[\attr{Winery}]\leq\val{Langhe}) \wedge 
			    (y[\attr{Year}]\leq\val{young})
\end{array}
\]
The above statements, when evaluated over the t-tuples in Figure~\ref{fig:concerts}, yield the following preferences, written as pairs of t-tuples in $\wpref$ (for the sake of clarity, for each preference we also show the statement used to derive it): 
\[
\begin{array}{rl}
P_1:& \, (a,b), (a,c), (a,e), (f,b), (f,c), (f,e) \\
P_2:& \, (b,a), (b,f), (c,a), (c,f)  \\
P_4:& \, (e,f) 
\end{array}
\]
Notice that $P_3$ yields no preference, since there is no wine from Asti's province in the t-relation in Figure~\ref{fig:concerts}.
\end{example}

In the rest of the paper, with the aim to simplify the notation, preference statements in the examples will be written with a compact syntax, by omitting variables and attributes' names, and separating with $\wpref$ the ``better'' part from the ``worse'' part. For instance, the above statement $P_4$ will be written as:
 
\[
\begin{array}{rl}
P_4 =  & \val{Langhe} \wedge \val{aged} \wpref \val{Langhe} \wedge \val{young}.
\end{array}
\]

\section{Operations on Preferences}\label{sec:propagation}

In this section we introduce two operators that can be applied to a preference relation, postponing to the next section the detailed analysis of the possible ways in which they can be combined. The two operators are: Transitive closure (\Top) and Specificity-based refinement (\Sop). 
Let $\wpref$ denote the initial preference relation; the resulting relation is indicated $\wpreft$ for \Top\ and $\wprefs$ for \Sop.
Multiple application of operators, e.g., first \Top\ and then \Sop, leads to the relation $(\wpreft)_\Sop$, which we compactly denote as $\wpref_{\Top\Sop}$. In general, for any \emph{sequence} $X\in\{\Top,\Sop\}^*$, $\wprefx$ is the preference relation obtained from the initial preference relation $\wpref$ by applying the operators in the order in which they appear in $\Xop$. Notice that $\wprefO\; =\; \wpref$, where $\Emptyseq$ denotes the empty sequence.

We describe the behavior of the two operators by means of suitable rewritings of a preference formula. Given a sequence $\Xop$ of operators, and an initial (input) formula $F(x,y)$ inducing the preference relation $\wpref$, $F^\Xop(x,y)$ denotes the rewriting of $F$ that accounts for the application of the $\Xop$ sequence, thus yielding $\wprefx$.

\newtheorem{problem}{Problem}

\subsection{Transitive Closure}\label{sec:closure}

Transitivity of $\wpref$, and consequently of $\pref$, is a basic requirement of
any sound preference-based system.
  If $\wpref$ is not transitive then $\pref$ might contain cycles, a fact that could easily lead either to empty or non-stable results, as the following example shows.
\begin{example}\label{ex:non-transitive-algorithm}
Consider the t-tuples in Figure~\ref{fig:concerts2}, in which both \val{Sbarbata} and \val{Molinara} are ros\'e wines and \val{Vogadori} is a winery in the \val{Valpolicella} wine region.
\begin{figure}[h]
\begin{center}\rm%
\begin{tabular}{|ccc|c}
\cline{1-3}
\textsf{Wine} &\textsf{Winery} &\textsf{Year}\\
\cline{1-3}
 Arneis & Correggia & 2019 &\ $g$\\
 Barolo & Laficaia &  2014 &\ $h$\\
 Sbarbata & Laficaia & 2019 &\ $\ell$\\
 Molinara & Vogadori & 2014 &\ $m$\\
\cline{1-3}
\end{tabular}
\end{center}
\caption{A set of wines for Example \ref{ex:non-transitive-algorithm}.}\label{fig:concerts2}
\end{figure}

From the preference statements in Example \ref{ex:pref-concerts}, we have $g \wpref h$ (through $P_1$) and $h \wpref \ell$ (through clause $P_4$). However, $g \not\wpref \ell$.
Assume now two additional preference statements
\[
\begin{array}{rl}
P_\alpha =  & \val{ros\'e} \wedge \val{young} \wpref \val{ros\'e} \wedge \val{aged},\\
P_\beta =  & \val{Valpolicella} \wpref \val{Roero},
\end{array}
\]
which, respectively, induce preferences $\ell\wpref m$ and $m \wpref g$.
Overall, since no other preferences hold, we have the non-transitive cycle of strict preferences $g \pref h$, $h \pref \ell$, $\ell \pref m$ and $m \pref g$.
So, for a t-relation $r=\{g, h, \ell, m\}$, we have $\Best_{\pref}(r)=\emptyset$.

Consider now $r' = \{g,h,\ell\}$, for which $\Best_{\pref}(r')=\{g\}$, and $r''=\{g,\ell,m\}$, for which $\Best_{\pref}(r'')=\{\ell\}$.
Although both $r'$ and $r''$ contain $g$ and $\ell$, the choice of which of these t-tuples is better than the other depends on the presence of other t-tuples (like $h$ and $m$), thus making the result of the $\Best$ operator unstable.
\end{example}

The transitive closure operator, denoted \Top, given an input preference relation $\wprefx$ yields the preference relation $\wprefxt$.
We remind that, as observed in Section~\ref{sec:preferencemodel}, the transitivity of $\wprefxt$ entails that of $\prefxt$.
The transitive closure $F^{\Xop\Top}$ of an ipf $F^\Xop$ with $n$ statements $P_1,\ldots,P_n$ is still a finite ipf that can be computed via Algorithm \ref{alg:transitive-closure}, along the lines described in~\cite{Ch03}. For the sake of conciseness, given a preference clause $C(x,y)$, we denote by $C^b(x)$ (resp.\ $C^w(y)$) the part of $C(x,y)$ given by the conjunction of the predicates involving variable $x$ (resp.\ $y$). Notice that $C(x,y) = C^b(x) \wedge C^w(y)$ holds.

In the main loop of the algorithm (lines (\ref{line:repeat-T})--(\ref{line:until-T})) we test the possibility of transitively combining two preference statements at a time (line (\ref{line:statements-T})), by considering each of their clauses (line (\ref{line:clauses-T})).
Since clauses are assumed to be non-contradictory, the test at line (\ref{line:if-T}), which can also be written as $C^b_m(t_1) \wedge C^w_m(t_2) \wedge C^b_q(t_2) \wedge C^w_q(t_3)$, reduces to checking if  $C^w_m(t_2) \wedge C^b_q(t_2)$ is satisfiable in $\domain$.
This can be done by checking whether no contradictory pair of predicates occurs in $C^w_m(t_2) \wedge C^b_q(t_2)$.
In particular, two predicates of the form $(x[A_i]\; \povi\; v_1)$ and $(x[A_i]\; \povi\; v_2)$, over the same attribute $A_i$ and using the same variable $x$, are contradictory if values $v_1$ and $v_2$ are different and have no common descendant in the taxonomy $V_i$ (Section~\ref{sec:experiments} further discusses how to check the existence of a common descendant). If the predicates are of the form $(x[A_i]\; \povi\; v_1)$ and $(x[A_i]\; \not\leq_{V_i}\;  v_2)$, then they are contradictory in case there is a path from $v_1$ to $v_2$ in $V_i$ (or $v_1=v_2$).
\begin{algorithm}[t]
\scalebox{.95}
   {
    \begin{minipage}{1.33\textwidth}
	\begin{enumerate}
	    \item[Input:] \emph{formula $F^\Xop = P_1 \vee \ldots \vee P_n$, taxonomies $T_1,\ldots,T_d$}. 
	    \item[Output:] \emph{$F^{\Xop\Top}$, the transitive closure of $F^\Xop$}. 
	    \item $F^{\Xop\Top} := F^{\Xop}$
	    \item \coderepeat \label{line:repeat-T}
	    \item \quad $newPref :=$ \codefalse
	    \item \quad \codeforeach\ ordered pair $(P_i,P_j)$, $P_i$ in $F^{\Xop\Top}$, $P_j$ in $F^{\Xop}$ \label{line:statements-T}
		    \item \quad \quad $P :=$ empty	    
		    \item \quad\quad \codeforeach\ ordered pair $(C_m,C_q)$, $C_m$ in $P_i$, $C_q$ in $P_j$ \label{line:clauses-T}
				\item \quad\quad \quad \codeif\ $\exists\ t_1,t_2,t_3 \in\domain$ s.t.\ $C_m(t_1,t_2) \wedge C_q(t_2,t_3) =$ \codetrue \label{line:if-T} 
					\item[] \quad\quad \quad\quad \codethen\ $P := P \vee (C^b_m(x) \wedge C^w_q(y))$		
		 
	    	\item \quad \quad \codeif\ $P \neq$ empty \codethen\ $F^{\Xop\Top} := F^{\Xop\Top} \vee P$, $newPref :=$ \codetrue
	    \item \codeuntil\ $newPref =$ \codefalse	\label{line:until-T}
	    \item \codereturn\ $F^{\Xop\Top}$
	\end{enumerate}	    
	\end{minipage}
    }
	\caption{\Top\ operator: Transitive closure of $F^\Xop$.}
	\label{alg:transitive-closure}
\end{algorithm}

The fact that the transitive closure is computed with respect to the (possibly infinite) domain $\domain$ of the t-tuples, and \emph{not} with respect to a (finite) t-relation $r$ of t-tuples, is quite standard for preference relations (see e.g., \cite{Ch03}), and has the advantage of yielding a relation $\wprefxt$ that does not change with $r$ and avoiding the problems discussed in Example~\ref{ex:non-transitive-algorithm}.

\begin{example}\label{ex:pref-concerts-T}
Continuing with Example \ref{ex:pref-concerts}, the transitive closure of $F$ is the formula $F^{\Top}$ that, among others, adds the following statements to $F$:
\[
\begin{array}{rll}
P_5 =& \val{Amarone} &\wpreft \val{red} \\
P_6 =& \val{Siena} &\wpreft \val{Langhe} \wedge \val{young}
\end{array}
\]
Statement 
$P_5(x,y)$ clearly follows from $P_2(x,z)$ and $P_1(z,y)$. More interesting is statement $P_6(x,y)$, obtained from $P_3(x,z)$ and $P_4(z,y)$. 
Since there exists at least one winery that is both in the \val{Asti} province and in the \val{Langhe} region (\val{Casorzo} is one of them), this allows $P_3(x,z)$ and $P_4(z,y)$ to be transitively combined. With reference to the t-tuples in Figure \ref{fig:concerts}, we then have $d \wpreft f$. 
\end{example}

After applying the $\Top$ operator, we simplify the formula as needed, and, in particular, we remove statements that are subsumed by other statements.
Similarly, we also simplify statements by removing contradictory clauses and clauses subsumed within the same statement.

\subsection{Specificity-based Refinement}\label{sec:specifity}

The most intriguing of our operators is the \emph{specificity-based refinement} \Sop. As it is also apparent from Example \ref{ex:pref-concerts}, \emph{conflicting preferences}, such as $(a,b)$ and $(b,a)$, may hold. Although these preferences are compatible with the given definition of preference relation, we argue that some of these conflicts need to be resolved in order to derive a preference relation that better represents the stated user preferences. To this end we resort to a \emph{specificity principle}, which we adapt from the one typically used in non-monotonic reasoning to solve conflicts. According to such a principle, a conclusion derived from a more specific antecedent overrides a conflicting (defeasible) inference based on a less specific antecedent, that is, more specific information overrides more generic information.

\begin{example}\label{ex:cycle-two}
In our working example, we have a generic preference for white wines over red wines.
With no contradiction with the generic preference, we might have a \emph{more specific} preference stating that a bottle of Amarone (a red wine) is superior to a bottle of Arneis (a white wine). 
In this case, the more specific preference would entail, among others, $b \wpref a$; yet, because of the more generic preference for white wines, we also have $a \wpref b$, thus $a$ and $b$ become indifferent.
However, giving the same importance to both preference statements contradicts the intuition, as the more specific preference should take precedence over the more generic one.
\end{example} 

The specificity principle we adopt for analyzing conflicting preferences is based on the \emph{extension} of preferences statements, i.e., on the set of pairs of t-tuples in $\domain$ for which a statement is true. 

\begin{definition}[Specificity principle]\label{def:specificity}
Let $\wprefx$ be a preference relation, and let $F^\Xop$ be the corresponding formula.
Let $P_i$ and $P_j$ be two preference statements in $F^\Xop$.
We say that $P_i$ is \emph{more specific} than $P_j$ if, for any pair of t-tuples $t_1,t_2 \in \domain$ such that $P_i(t_1,t_2)$ is true, then $P_j(t_2,t_1)$ is also true, and the opposite does not hold.
\end{definition}

From Definition \ref{def:specificity} we can immediately determine how a less specific statement has to be rewritten so as to solve conflicts.

\begin{lemma}\label{lem:implication}
A preference statement $P_i(x,y)$ is more specific than $P_j(y,x)$ iff $P_i(x,y)$ implies $P_j(y,x)$ (written $P_i(x,y) \impl P_j(y,x)$) and the opposite does not hold.\footnote{The hypothesis that $P_j(y,x)$ does not imply $P_i(x,y)$ excludes the case of \emph{opposite preference statements} (e.g., white is better than red, and red is better than white), to which the \Sop\ operator clearly does not apply.}
If $P_j(y,x)$ is replaced by $P'_j(y,x) = P_j(y,x) \wedge \neg P_i(x,y)$, then $P_i$ and $P'_j$ do not induce any conflicting preferences.
\end{lemma}

Checking whether $P_i(x,y)$ implies $P_j(y,x)$ amounts to checking whether $P_i(x,y) \land \lnot P_j(y,x)$ is false, i.e.,
every clause in the resulting formula is contradictory (contradictions can be checked as described for $\Top$).

The \Sop\ operator, whose behavior is defined by Algorithm \ref{alg:specificity}, removes from the preferences induced by a formula $F^{\Xop}$ all those that are conflicting and less specific. 

Notice that, after a first analysis of the existing implications  among the statements (line (\ref{line:implset-S})) and the rewriting of the implied statements (line (\ref{line:rewrite-S})), the analysis needs to be repeated, since new implications might arise. 
For instance, let $F^\Xop = P_1 \vee P_2 \vee P_3$, with $P_1(y,x) \impl P_2(x,y)$ being the only implication. After rewriting $P_2(x,y)$ into $P'_2(x,y) = P_2(x,y) \wedge \neg P_1(y,x)$, it might be the case that $P'_2(x,y) \impl P_3(y,x)$, thus $P_3$ needs to be rewritten.

Although multiple rounds might be needed, Algorithm \ref{alg:specificity} is guaranteed to terminate. 
Indeed, if $P_i(x,y) \impl P_j(y,x)$, and $P_j(y,x)$ is consequently replaced by $P'_j(y,x) = P_j(y,x) \wedge \neg P_i(x,y)$, the two statements $P_i$ and $P'_j$, as well as their possible further rewritings, have \emph{disjoint} extensions, and therefore will not interact anymore in the rewriting process. Since the number of statements is finite, so is the number of rewritings, which ensures that the algorithm will eventually stop.

\begin{algorithm}[t] 
\scalebox{.95}
   {
    \hspace{-0.5cm}
    \begin{minipage}{1.33\textwidth}
	\begin{enumerate}
	    \item[Input:] \emph{formula $F^\Xop = P_1 \vee \ldots \vee P_n$, taxonomies $T_1,\ldots,T_d$}. 
	    \item[Output:] \emph{$F^{\Xop\Sop}$, the specificity-based refinement  of $F^\Xop$}. 
	    \item \coderepeat 
	    \item \quad $newRound :=$ \codefalse
	    \item \quad \codeforeach\ statement $P_i$ 
		    \item \quad \quad $Impl(P_i) := \{P_j | P_j(y,x) \impl P_i(x,y) \wedge P_i(x,y) \not\impl P_j(y,x)\}$ \label{line:implset-S}
		    \item \quad \quad \codeif\ $Impl(P_i) \neq \emptyset$
				\codethen\
			\item[] \quad\quad\quad	$newRound :=$ \codetrue, $P'_i ;= P_i$
			\item[] \quad\quad\quad \codeforeach\ $P_j \in Impl(P_i)$
			\item[]   \quad\quad\quad\quad
			$P'_i(x,y) := P'_i(x,y) \wedge \neg P_{j}(y,x)$ \label{line:rewrite-S}	
	    \item \quad \codeif\ $newRound$ \codethen\ $P_i := P'_i, i=1,..,n$
		\item \codeuntil\ $newRound =$ \codefalse
	    \item \codereturn\ $F^{\Xop\Sop} = P_i \vee \ldots \vee P_n$
	\end{enumerate}	    
	\end{minipage}
    }
	\caption{\mbox{\Sop\ operator: Specificity-based refinement of $F^\Xop$.}}
	\label{alg:specificity}
\end{algorithm}

Here too, we simplify the formula resulting from the rewritings according to the same principles used for the $\Top$ operator.

\begin{example}\label{ex:pref-concerts-TS}
Continuing with Example \ref{ex:pref-concerts-T}, the application of the \Sop\ operator amounts to rewriting formula $F^{\Top}$ by replacing the clause $P_1(x,y)$ with $P_1(x,y)\wedge \neg P_2(y,x)$, since $P_2(y,x) \impl P_1(x,y)$. 
This, after distributing $\neg$ over the two predicates in $P_2$ and simplifying, leads to the new clause:
\[
\begin{array}{rll}
P_7 =& \val{white} &\wpref_{\Top\Sop} \val{red} \wedge \neg\val{Amarone}. 
\end{array}
\]
The preferences that were derived from $P_1$ can be seen in Example \ref{ex:pref-concerts}; we repeat them for the sake of clarity:
\[
\begin{array}{rl}
P_1:& \, (a,b), (a,c), (a,e), (f,b), (f,c), (f,e).
\end{array}
\]
Among them, $(a,b),(a,c),(f,b)$, and $(f,c)$ do not satisfy $P_7(x,y)$, since both $b$ and $c$ refer to \val{Amarone}. Thus,
$P_7: \, (a,e), (f,e).$
\end{example}

It is relevant to observe that the application of the \Sop\ operator always leads to smaller (i.e., cleaner) results.
For instance, considering t-relation $r$ in Figure~\ref{fig:concerts} and input preference statements $P_1$ and $P_2$ from Example~\ref{ex:pref-concerts}, we have $\Best_{\pref}(r)=\{a, b, c, d, f\}$, whereas $\Best_{\pref_\Sop}(r)=\{b, c, d\}$.

\begin{restatable}{lemma}{thmBestXS}\label{lem:BestXS}%
For any t-relation $r$ and any preference relation $\wpref_{\Xop}$ we have $\Best_{\pref_{\Xop\Sop}}(r) \subseteq \Best_{\pref_{\Xop}}(r)$.
\end{restatable}

\section{Minimal-Transitive Sequences}\label{sec:semantics}

In this section we analyze the effect of performing the operations described in the previous section, and prove some fundamental properties of the obtained preference relations.
After introducing the basic properties and main desiderata in Section~\ref{sec:basic-properties-sequences}, we explore the space of possible sequences in Section~\ref{sec:space-sequences} and, as a major result, we show that, out of infinitely many candidates, only a finite number of sequences needs to be considered. Finally, in Section~\ref{sec:min-transitivity} we identify the only two sequences that meet all our requirements.

\subsection{Basic properties}\label{sec:basic-properties-sequences}
In order to clarify the relationships between the results of the different operations, we introduce the notions of equivalence and containment between sequences of operators.
\begin{definition}[Equivalence and containment]
	Let $\Xop,\Yop\in\{\Top,\Sop\}^*$; $\Xop$ is \emph{contained} in $\Yop$,
	denoted $\Xop\sqsubseteq\Yop$, if for every initial preference relation $\wpref$, $\wpref_{\Xop}\subseteq\wpref_{\Yop}$; $\Xop$ and $\Yop$ are \emph{equivalent}, denoted $\Xop\equiv\Yop$, if both $\Xop\sqsubseteq\Yop$ and $\Yop\sqsubseteq\Xop$.
\end{definition}

Among the basic properties of our operators, we observe that $\Top$ and $\Sop$ are idempotent, $\Top$ is monotone and cannot remove preferences, while $\Sop$ cannot add preferences.
In addition, the preference relation obtained after applying 
$\Top$ on the initial preference relation $\wpref$ is maximal, in that it includes all other relations obtained from $\wpref$ by applying $\Top$ and $\Sop$ in any way.

\begin{restatable}{theorem}{thmBasicProperties}\label{thm:basic-properties}%
	Let $\Xop,\Yop\in\{\Top,\Sop\}^*$, with $\Xop\sqsubseteq\Yop$. Then:
\begin{align}
\Xop\Top\Top&\equiv\Xop\Top	&\Xop\Sop\Sop&\equiv\Xop\Sop&&\mbox{idempotence}\\
\Xop\Top&\sqsubseteq\Yop\Top&&&&\mbox{monotonicity}\\
\Xop&\sqsubseteq\Xop\Top&\Xop\Sop&\sqsubseteq\Xop&&\mbox{inflation / deflation}\\
\Xop&\sqsubseteq\Top&&&&\mbox{maximality}
\end{align}
\end{restatable}

We now focus on those sequences, that we call complete, that include both $\Top$ and $\Sop$, since their corresponding operations are both part of the our requirements.
In particular, 
transitivity of the obtained \emph{strict} preference relation is at the core of the computation of the Best ($\Best$) operator, as shown in Example~\ref{ex:cycle-two}.
To this end, we characterize as transitive those sequences that entail such a transitivity.

\begin{definition}[Complete and transitive sequence]\label{def:ct-sequence}
A sequence $\Xop\in\{\Top,\Sop\}^*$ is \emph{complete} if $\Xop$ contains both $\Top$ and $\Sop$; 
$\Xop$ is \emph{transitive} if, for every initial preference relation $\wpref$, 
$\pref_\Xop$ is transitive.
\end{definition}
Eventually, our goal is to drop conflicting and less specific preferences while preserving transitivity.
To this end, we add minimality with respect to $\sqsubseteq$ as a desideratum.
In particular, we want to determine the so-called \emph{minimal-transitive} sequences, i.e., those that are minimal among the transitive sequences. As it turns out, all such sequences are also complete.
\begin{definition}[Minimal-transitive sequence]
Let $\seqSet$ be a set of sequences; $\Xop\in\seqSet$ is \emph{minimal} in $\seqSet$ if there exists no other sequence $\Yop\in\seqSet$, $\Yop\not\equiv\Xop$ such that $\Yop\sqsubseteq\Xop$.
A \emph{minimal-transitive} sequence is a sequence that is \emph{minimal} in the set of 
\emph{transitive} sequences.
\end{definition}

\subsection{The space of possible sequences}\label{sec:space-sequences}
We now chart the space of possible sequences so as to understand the interplay between completeness, transitivity and minimality.

We start by observing that any sequence with consecutive repetitions of the same operator is equivalent, through idempotence, to a shorter sequence with no such repetitions; for instance, $\Top\Sop\Sop$ is equivalent to $\Top\Sop$. Since sequences with repetitions play no significant role in our analysis, we shall henceforth disregard them.

Clearly, every sequence is contained in $\Top$, due to its maximality. Other containment relationships follow from inflation of $\Top$ and deflation of $\Sop$. Further relationships come from the following result, stating that adding $\Sop\Top$ (i.e., removing conflicts and then transitively closing the resulting preference formula) to a sequence ending with $\Top$ cannot introduce any new preference.

\begin{restatable}{lemma}{thmXTcontainsXTST}\label{thm:XT-contains-XTST}%
	Let $\Xop\in\{\Top,\Sop\}^*$. Then $\Xop\Top\Sop\Top\sqsubseteq\Xop\Top$.
\end{restatable}

Lemma~\ref{thm:XT-contains-XTST} induces two chains of inclusions, namely:
\begin{align}
\ldots\sqsubseteq\Top\Sop\Top\Sop\Top&\sqsubseteq\Top\Sop\Top\sqsubseteq\Top\label{eq:tchain}\\
\ldots\sqsubseteq\Sop\Top\Sop\Top\Sop\Top&\sqsubseteq\Sop\Top\Sop\Top\sqsubseteq\Sop\Top.\label{eq:schain}
\end{align}

In addition to that, the following result seems to suggest that the longer sequences in the above chains are preferable, since they lead to larger sets of strict preferences ($\pref$), which, as was observed in Section~\ref{sec:preferencemodel}, correspond to smaller (i.e., cleaner) results for the Best $\Best$ operator.

\begin{restatable}{proposition}{thmXTSTstrictContainsXTstrict}\label{thm:XTSTstrict-contains-XTstrict}%
	Let $\Xop\in\{\Top,\Sop\}^*$. Then, for any initial preference relation $\wpref$, we have $\pref_{\Xop\Top}\subseteq\pref_{\Xop\Top\Sop\Top}$.
\end{restatable}

There are, evidently, infinitely many sequences in the chains~\eqref{eq:tchain} and~\eqref{eq:schain} and, more generally, in $\{\Top,\Sop\}^*$. However, for any given initial preference formula, a counting argument on the number of formulas obtainable through the operators allows us to restrict to only a finite amount of sequences.
Moreover, it turns out that the repeated application of a $\Top\Sop$ suffix does not change the semantics of a sequence, so we can apply it just once and disregard all other sequences.
\begin{restatable}{lemma}{thmTSrepeated}\label{thm:TS-repeated}%
Let $\Xop\in\{\Top,\Sop\}^*$. Then $\Xop\Top\Sop\equiv\Xop\Top\Sop\Top\Sop$.
\end{restatable}

An immediate consequence of this result is that, through elimination of consecutively repeated operators via idempotence and of consecutively repeated $\Top\Sop$ sub-sequences via Lemma~\ref{thm:TS-repeated},
we can restrict our attention to a set of just eight sequences, because any sequence is equivalent to one of those.

\begin{restatable}{theorem}{thmFiniteRepresentatives}\label{thm:finite-representatives}%
Let $\Xop\in\{\Top,\Sop\}^*$. Then $\exists\Yop\in\{\Emptyseq,\Top,\Sop,\Top\Sop,\Sop\Top,\Top\Sop\Top,$ $\Sop\Top\Sop,\Sop\Top\Sop\Top\}$ such that $\Xop\equiv\Yop$.
\end{restatable}

\begin{figure}[h]
\centering
\includegraphics[width=0.4\columnwidth]{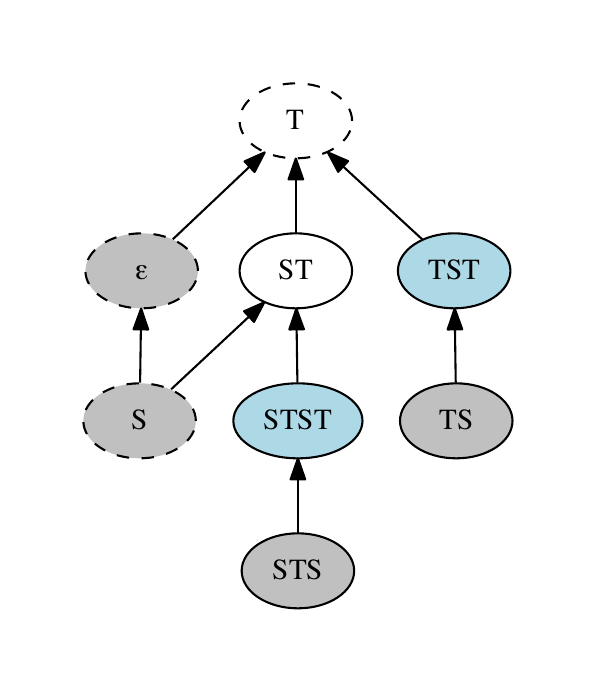}
\caption{A transitively reduced graph showing containment between sequences. Dashed border for incomplete sequences; grey background for non-transitive sequences; blue background for minimal-transitive sequences. All containment relationships are strict.
 }\label{fig:contenimentiSequenze}
\end{figure}

Figure~\ref{fig:contenimentiSequenze} shows a (transitively reduced) graph whose nodes correspond to the eight sequences mentioned in Theorem~\ref{thm:finite-representatives} and whose arcs indicate containment.
Thanks to the theorem, we have narrowed the space of possible sequences to analyze from an infinite set $\{\Top,\Sop\}^*$ to just these eight sequences.

\subsection{Minimality and transitivity}\label{sec:min-transitivity}
Now that we have restricted our scope to a small set of representative sequences, we can discuss minimality and transitivity in detail, so as to eventually detect minimal-transitive sequences.
Note that incomplete sequences can be immediately ruled out of our analysis: it is straightforward to show that $\Sop$ is not transitive, $\Top$ is not minimal (it is indeed maximal) and $\Emptyseq$ is neither.

\medskip

\textbf{Minimality.}
Generally, any complete sequence not ending with $\Sop$ is non-minimal, in that it may contain conflicting preferences (possibly introduced by $\Top$) that turn out to be in contrast with other, more specific preferences.
We exemplify this on $\Sop\Top$.
In the examples to follow, we shall refer to t-tuples with a single attribute on a single taxonomy about time.
\begin{example}\label{ex:ST-non-minimal}
Let $F$ consist of $P_1$ and the more specific $P_2$:
\[
\begin{array}{rrcl}
P_1 =   \B		\wpref \A ,
&P_2 =  \auno		\wpref \bdue.
\end{array}
\]
By specificity, in $F^\Sop$, 
$P_1$ is replaced by the statement $P_3$ consisting of two clauses (grouped by curly brackets):
\[
\begin{array}{r}
P_3 = \left\{
	\begin{array}{rcl}
	    \B &\!\!\!\wpref\!\!\!& \A  \land \lnot \auno\\
	    \B \land \lnot  \bdue &\!\!\!\wpref\!\!\!& \A
	\end{array}
	\right.
\end{array}
\]
In $F^{\Sop\Top}$, the clauses in $P_3$ transitively combine into $P_1$ again,
since, e.g., the value $\abdue$ is below $\A$ but not $\auno$ and below $\B$ but not $\bdue$; therefore $\bdue\wpref_{\Sop\Top}\auno$ holds.
However, in $F^{\Sop\Top\Sop}$, $P_1$ is again replaced by $P_3$, so that $\bdue\not\wpref_{\Sop\Top\Sop}\auno$, which shows that $\Sop\Top$ is not minimal.
\end{example}

All the containments indicated in Figure~\ref{fig:contenimentiSequenze} are strict, as can be shown through constructions similar to that of Example~\ref{ex:ST-non-minimal}, so no sequence ending with $\Top$ is minimal in $\{\Top,\Sop\}^*$.

\begin{restatable}{lemma}{thmTnonminimal}\label{thm:T-non-minimal}%
Let $\Xop\in\{\Top,\Sop\}^*$. Then $\Xop\Top$ is not minimal in $\{\Top,\Sop\}^*$.
\end{restatable}

\textbf{Transitivity.}
Transitivity is certainly achieved for any sequence ending with $\Top$: any relation $\wpref_{\Xop\Top}$ is transitive by construction, which entails transitivity of $\pref_{\Xop\Top}$.
However, the following result shows that, in the general case, no sequence ending with $\Sop$ is transitive.
\begin{restatable}{lemma}{thmSnontransitive}\label{thm:S-non-transitive}%
Let $\Xop\in\{\Top,\Sop\}^*$. Then $\Xop\Sop$ is not transitive.
\end{restatable}

\textbf{Minimal-transitive sequences.}
\sloppy
As a consequence of Lemmas~\ref{thm:T-non-minimal} and~\ref{thm:S-non-transitive}, we can state a major result, showing that transitivity and minimality in $\{\Top,\Sop\}^*$ are mutually exclusive.
\begin{restatable}{theorem}{thmtransitivityvsminimality}\label{thm:transitivity-vs-minimality}%
	No sequence is both transitive and minimal in $\{\Top,\Sop\}^*$.
\end{restatable}

Moreover, we observe that all complete sequences starting with $\Sop$ are incomparable (i.e., containment does not hold in any direction) with those starting with $\Top$,
  as stated below (also refer to~Figure~\ref{fig:contenimentiSequenze}).

\begin{restatable}{theorem}{thmincomparable}\label{thm:incomparable}%
Let $\Xop\in\{\Top\Sop,\Top\Sop\Top\}$ and $\Yop\in\{\Sop\Top,\Sop\Top\Sop,\Sop\Top\Sop\Top\}$. Then $\Xop\not\sqsubseteq\Yop$ and $\Yop\not\sqsubseteq\Xop$.
\end{restatable}

This property is shown for $\Top\Sop$ and $\Sop\Top\Sop$ in the next example.
\begin{example}\label{ex:TS<>STS}
Let $F$ consist of the following statements:
\[
\begin{array}{rll}
P_1 = \summer \wpref \spring,
&P_2 = \jultwentyone \wpref \jun,
&P_3 = \may \wpref \jul.
\end{array}
\]
Then $F^\Top$ includes $P_1$, $P_3$ and the following 4 statements:
\[
\begin{array}{rrcllrrcll}
\!\!\!\!P_4= &\!\!\!\!\summer& \!\!\!\!\wpref\!\!\!\!& \jul& \!\!\!\!\mbox{ ($P_1+P_3$}),
\!\!\!\!&P_5=&\!\!\!\!\may& \!\!\!\!\wpref\!\!\!\!& \spring& \!\!\!\!\mbox{ ($P_3+P_1$}),\\
\!\!\!\!P_6= &\!\!\!\!\may& \!\!\!\!\wpref\!\!\!\!& \jun& \!\!\!\!\mbox{ ($P_3+P_2$}),
\!\!\!\!&P_7=&\!\!\!\!\summer& \!\!\!\!\wpref\!\!\!\!& \jun& \!\!\!\!\mbox{ ($P_1+P_6$}),
\end{array}
\]
while $P_2$ is removed, as it is redundant with respect to $P_7$.
No statement in $F^\Top$ is more specific than $P_4$, so $P_4$ is in
$F^{\Top\Sop}$ and, e.g., $\jultwentyone \wpref_{\Top\Sop} \julten$
 holds.
In $F^\Sop$, instead, $P_1$ (less specific than $P_3$) is replaced by
\[P_8 = \left\{
	\begin{array}{rcl}
	     \summer &\!\!\!\!\wpref\!\!\!\!& \spring \land \lnot \may\\
	    \summer \land \lnot \jul &\!\!\!\!\wpref\!\!\!\!& \spring
	\end{array}\right.
\]
So, now, by combining $P_8$ (instead of $P_1$) and $P_3$, in $F^{\Sop\Top}$ we do not obtain
$P_4$ and then
$\jultwentyone \not\wpref_{\Sop\Top\Sop} \julten$.
With this, $\Top\Sop\not\sqsubseteq\Sop\Top\Sop$.

For the other non-containment, consider that,
in $F^{\Sop\Top}$, $P_2$ combines with $P_8$ into the following statement:
\[P_{9} = \jultwentyone \wpref \spring,\]
so that $\jultwentyone \wpref_{\Sop\Top} \may$ holds.
No statement in $F^{\Sop\Top}$ is more specific than $P_{9}$, so $\jultwentyone \wpref_{\Sop\Top\Sop} \may$ also holds.
Instead, $\jultwentyone \not\wpref_{\Top\Sop} \may$, since $F^{\Top\Sop}$ is as $F^\Top$, but with $P_8$ instead of $P_1$.
Therefore $\Sop\Top\Sop\not\sqsubseteq\Top\Sop$.
\end{example}

The notion of minimal-transitive sequence captures the fact that transitivity cannot be waived, since we are indeed looking for the minimal sequences among those that are both complete and transitive.
Only three sequences are both complete and transitive: $\Sop\Top$, $\Top\Sop\Top$ and $\Sop\Top\Sop\Top$, the first of which contains the last one and is therefore not minimal.
The remaining two sequences are transitive, incomparable by Theorem~\ref{thm:incomparable}, and, therefore, minimal in the set of complete and transitive sequences, i.e., $\Top\Sop\Top$ and $\Sop\Top\Sop\Top$ are minimal-transitive sequences.

\begin{restatable}{theorem}{thmMCT}\label{thm:MCT}%
	The only minimal-transitive sequences are $\Top\Sop\Top$ and $\Sop\Top\Sop\Top$.
\end{restatable}

As observed in Theorem~\ref{thm:incomparable}, the sequence $\Sop\Top\Sop\Top$, which removes less specific conflicting preferences before computing the first transitive closure, does not in general entail a set of preferences included in those induced by $\Top\Sop\Top$. We shall further characterize the behavior of these two sequences in Section~\ref{sec:algorithm}, from a theoretical point of view, and, experimentally, in Section~\ref{sec:experiments}.

We also observe that the result of Theorem~\ref{thm:transitivity-vs-minimality} is inherent 
and that no finer granularity in the interleaving of $\Top$ and $\Sop$ (e.g., by making $\Sop$ resolve one conflict at a time instead of all together) would remove this limitation: as Example~\ref{ex:ST-non-minimal} shows, the presence of one \emph{single} preference ($\bdue\wpref\auno$) is sufficient to make the relation transitive but not minimal, and its absence to make it minimal but not transitive.
The atomicity of this conflict is enough to conclude that 
it is unavoidable and that no method whatsoever (not just those based on the $\Top$ and $\Sop$ operators) could solve it.

\section{Computing the Best Results}\label{sec:algorithm}

\subsection{\!\!\!Worst-case difference between $\Top\Sop\Top$ and $\Sop\Top\Sop\Top$}

As shown in Theorem \ref{thm:incomparable}, the two minimal-transitive semantics are incomparable, thus there will be t-relations $r$ and initial preference relations $\wpref$ for which the best results delivered by the two semantics will differ. A legitimate question is: How much can these results be different?
In order to answer this question we consider the \emph{maximum} value of the cardinality of the difference of the results delivered by the two minimal-transitive semantics over all t-relations with $n$ t-tuples and over all input preference relations $\wpref$. To this end, let us define, for any two sequences $\Xop$ and $\Yop$:
\begin{align*}
\DiffBest(\Xop,\Yop,n) = 
\max_{\wpref,\; |r| = n} \left\{ |\Best_{\pref_{\Xop}}(r)- \Best_{\pref_{\Yop}}(r)| \right\} 
\end{align*}
as the worst-case difference in the results delivered by \Xop\ with respect to those due to \Yop, for any given cardinality of the target t-relation $r$.
We can prove the following:
\begin{restatable}{theorem}{thmdiffBest}\label{thm:diffbest}
We have both
$\DiffBest(\Top\Sop\Top, \Sop\Top\Sop\Top,n) = \Theta(n)$ and
$\DiffBest(\Sop\Top\Sop\Top, \Top\Sop\Top,n) = \Theta(n)$. 
\end{restatable}
From a practical point of view, Theorem \ref{thm:diffbest} shows that there is no all-seasons minimal-transitive semantics.
Furthermore, 
there can be cases (used in the proof of the theorem) in which the number of best results from any of the two semantics is comparable to $n$, whereas the other semantics returns $\Oh(1)$ t-tuples. In Section~\ref{sec:experiments} we will experimentally investigate the actual difference of results delivered by the two minimal-transitive semantics. 

\subsection{A heuristics for computing the best results}

In order to compute the best results according to the formula
$F^{\Xop}$ we adopt the well-known BNL algorithmic pattern~\cite{BKS01}.
We remind that BNL-like algorithms have worst-case quadratic complexity, although in practice they behave almost linearly~\cite{DBLP:journals/vldb/GodfreySG07}.
Remind also that, according to Equation~\eqref{eq:strict-pref}, given a preference formula $F^{\Xop}(x,y)$ defining \emph{weak} preferences, the corresponding \emph{strict} preferences are those induced by the formula $F_\pref^{\Xop}(x,y) = F^{\Xop}(x,y) \wedge \neg F^{\Xop}(y,x)$.

The t-tuples that do not match any side of any clause in the preference formula correspond to those objects that the formula does not talk about and that can, thus, be considered irrelevant. 
As recognized in the germane literature~\cite{DBLP:conf/icde/GeorgiadisKCNS08}, such objects are of little interest and, in the following, we shall therefore compute $\Best$ so as to only include \emph{relevant} t-tuples (i.e., those that satisfy either side of at least one clause of $F^\Xop$, thus of $F_\pref^{\Xop}$ as well). 

The algorithm keeps the current best t-tuples in the $Best$ set. When a new t-tuple $t$ is read, and $t$ is found to be relevant, $t$ is compared to the tuples in $Best$. 
Given $t' \in Best$, if $t' \pref_\Xop t$ then $t$ is immediately discarded. Conversely,
$t$ is added to $Best$ and all t-tuples $t'\in Best$ such that $t \pref_\Xop t'$ are removed from the $Best$ set.
Eventually, we have $\Best_{\pref_\Xop}(r) = Best$.

An improvement to this basic scheme is to pre-sort the t-relation so that the t-tuples matching the left side of a clause and corresponding to lower-level values in the taxonomies come first. The rationale is that lower-level values are likely associated with a smaller amount of t-tuples, so that a smaller $Best$ partial result can be found before scanning large amounts of data. Furthermore, such t-tuples are likely to be preferred to many others, in particular when specificity is a concern.
More in detail, we scan $r$ and, for each relevant t-tuple $t$ (irrelevant t-tuples are immediately discarded) we compute a \emph{height index}, $hi(t)$, as follows:
For any clause $C(x,y) = C^b(x) \wedge C^w(y)$ such that $C^b(t)$ holds,
we consider the ``height'' of each value $v$ occurring in the clause, computed as the distance of $v$ from the leaves of its taxonomy.\footnote{In case of non-functional taxonomies, in which a node may have more than one parent, we take the minimum distance.} Then, the minimum height over predicates in $C^b(t)$ and over all other matching clauses is used as value of $hi(t)$, and t-tuples are sorted by increasing height index values; conventionally, when $t$ matches no clauses, we set $hi(t)=\infty$.

\begin{example}
Consider a formula $F= P_1 \lor P_2$, where $P_1$ and $P_2$ are taken from Example~\ref{ex:pref-concerts}.
Then, we have $F^{\Sop\Top\Sop\Top}=P_3 \lor P_2 \lor P_4$, where $P_3= \val{white} \wpref \val{red} \land \lnot \val{Amarone}$ and $P_4=\val{Amarone} \wpref \val{red} \land \lnot \val{Amarone}$.
Out of the t-tuples in Figure~\ref{fig:concerts}, $d$ is irrelevant, while $e$ does not match any clause, and thus $hi(e)=\infty$.
Wines $a$ and $f$ match $\val{white}$ in $P_3$, which has height $1$ (see Figure~\ref{fig:vini-colori}), so $hi(a)=hi(f)=1$, while $b$ and $c$ match $\val{Amarone}$ in both $P_2$ and $P_4$, with $hi(b)=hi(c)=0$. Thus, $b$ and $c$ come before $a$ and $f$ in the ordering, and $e$ is last.
\end{example}

\section{Experiments}\label{sec:experiments}

In this section, we consider from a practical point of view the sequences of operators $\Top$, $\Top\Sop\Top$, and $\Sop\Top\Sop\Top$, discussed in the previous sections.
The main goals of the experimental study are:
(i) to understand the impact of the rewriting process on the overall query execution time and how this depends on the specific sequence at hand;
(ii) to assess the effect of minimal-transitive sequences on (the cardinality of) the results of the $\Best$ (Best) operator;
(iii) to compare overall execution times incurred by minimal-transitive sequences with respect to baseline strategies in which either no rewriting occurs or only the transitive closure of the input formula is computed;
(iv) to measure the effects of the heuristics presented in Section~\ref{sec:algorithm}.
In particular, we study how efficiency and effectiveness are affected by 
taxonomy's size and morphology,
dataset size, number of attributes, and 
number and type of preferences.
The relevant parameters used in our analysis are summarized in Table~\ref{tab:operating_parameters}.

In summary, we show that:
the rewriting due to the minimal-transitive sequences $\Top\Sop\Top$ and $\Sop\Top\Sop\Top$ incurs a low overhead across all tested scenarios;
such sequences are effective both in reducing the cardinality of $\Best$ and in achieving substantial speedup with respect to baseline strategies, and that the speedup is further incremented when adopting our heuristics.

\newcommand{\fanout}{f}
\newcommand{\depth}{\delta}
\newcommand{\flipkart}{{\tt flipkart}\xspace}
\newcommand{\usedcars}{{\tt UsedCars}\xspace}
\newcommand{\level}{\ell}
\newcommand{\attributes}{d}
\newcommand{\clauses}{c}
\newcommand{\size}{N}

\begin{table}[h]
	\centering
	\caption{Operating parameters for performance evaluation (defaults, when available, are in bold).}
\scalebox{1}{
  		\begin{tabular}{|l|l|}
    		\hline
    			Full name							& Tested value \\
			\hline
				Taxonomy's depth $\depth$			& 2, 3, 4, 5, \textbf{6}, 7, 8, 9, 10\\
				Taxonomy's fanout $\fanout$			& 2, 3, 4, \textbf{5}, 6, 7, 8, 9, 10\\
				Synthetic taxonomy's kind  & \textbf{regular}, random, scale-free\\
				\# of attributes $\attributes$		& \textbf{1}, 2, 3, 4, 5 \\
				\# of input clauses $\clauses$		& \textbf{2}, 4, 6, 8, 10\\
				\# of maximal values  		& 2, 4, \textbf{6}, 8, 10\\
				Type of preferences					& \textbf{conflicting}, contextual \\
				Dataset size $\size$				& \textbf{10K}, 50K, 100K, 500K, 1M \\
			\hline
  		\end{tabular}
		}
	\label{tab:operating_parameters}
\end{table}

\subsection{Taxonomies, datasets, and preferences}
We use two families of taxonomies: synthetic and real taxonomies.

We run our tests on three kinds of synthetic taxonomies: regular, random and scale-free.
A regular taxonomy is generated as
a forest of $\fanout$ (``fanout'') rooted trees
consisting of $\depth$ levels
 and $\fanout$ children for each internal node. The total number of nodes is therefore $\sum_{i=1}^\depth \fanout^{i}$, i.e., $\frac{\fanout (\fanout^\depth-1)}{\fanout-1}$.
A random taxonomy is generated as in the previous case, but the fanout of each node is Poisson distributed with an average of $\fanout$.
The default values for $\fanout$ and $\depth$ are chosen to match the size of the real taxonomies used in the experiments
(15-20K nodes).
Finally, a scale-free taxonomy targets the same number of nodes, but following a power-law distribution (which is observed to be a recurrent structure, e.g., in the Semantic Web; see~\cite{DBLP:journals/is/TheoharisGC12,DBLP:journals/tkde/TheoharisTKC08}), for the fanout. Scale-free taxonomies generated this way (with reasonable exponents around 2.7) are typically very deep (between 30 and 60 levels).
All synthetic taxonomies are functional by construction, i.e., every node has exactly one parent.
Synthetic datasets of various sizes are generated by drawing values uniformly at random from a different taxonomy for each attribute.

We adopt two real taxonomies and datasets: \flipkart\footnote{\url{https://www.flipkart.com}} and \usedcars\footnote{\url{https://www.kaggle.com/austinreese/craigslist-carstrucks-data}}.
The former lists product categories 
of various kinds
and consists of 15,236 nodes (of which 12,483 leaf categories) and 15,465 arcs spread throughout $10$ levels.
This taxonomy is non-functional, in that there exist nodes with more than one parent, i.e., some products belong to more than one category.
Product info is available as a t-relation consisting of 19,673 t-tuples that also include original price, discounted price, and user rating, rendered here as attributes associated with a ``flat'' taxonomy with 
three values (e.g., ``high'', ``medium'', ``low'').
\usedcars features a large collection of used vehicles for sale consisting, after cleaning, of 232,470 t-tuples including, among others, price range (as a flat taxonomy) and model. Models are organized in a functional taxonomy, with 14,588 nodes and 14,540 leaves, over three levels (besides model name and make, we obtained country information via the Car Models List DB\footnote{\url{https://www.teoalida.com/cardatabase/car-models-list}}).

The study of the best taxonomy representation in the general case is orthogonal with respect to the problems we study in this paper (see, e.g.,~\cite{DBLP:conf/sigmod/JinXRW08}).
However, given the taxonomies we deal with, it is convenient to precompute all paths in order to speed up all taxonomy-based computations, e.g., establishing when a value is more specific than another.

For our experiments, we consider two common types of preferences, discussed below: conflicting preferences
and contextual preferences.
We omit the results concerning other common types of preferences, as their behavior is not essentially different.

A pair of \emph{conflicting preference} statements has the following form:
\[
P_1= v_1  \wpref v_2, \quad
P_2= v'_2 \wpref v_1,
\]
where $v_1$ and $v_2$ are
maximal
 values (i.e., tree roots)
  of the same taxonomy $T_i$
   and $v'_2\povi v_2$.
Clearly, $P_2$ is 
more specific than $P_1$.

The second kind of preferences, used for experiments on multi-attribute relations, are pairs of conflicting \emph{contextual preferences}, i.e., conflicting preferences applied to one attribute, in which the other attributes are used to establish a sort of ``context'' of applicability.
A pair of contextual preferences is of the following form:
\[
\begin{array}{rrcllrrcll}
 P_1=& v^{(1)}_1  \land v^{(2)} \land \ldots \land v^{(\attributes)} & \!\!\!\wpref\!\!\! & v^{(1)}_2 \land v^{(2)} \land \ldots \land v^{(\attributes)},\\
 P_2=& v'^{(1)}_2 \land v^{(2)} \land \ldots \land v^{(\attributes)} & \!\!\!\wpref\!\!\! & v^{(1)}_1 \land v^{(2)} \land \ldots \land v^{(\attributes)},
\end{array}
\]
where the $(i)$ superscript denotes values from taxonomy $T_i$, $v^{(1)}_1$ and $v^{(1)}_2$ are maximal in $T_1$, and $v'^{(1)}_2\povi v^{(1)}_2$. Note that, when there are $\attributes=1$ attributes, this is just a pair of conflicting preferences.
For real data, flat taxonomies are used for context attributes.
An example of contextual preference is given by statement $P_4$ in Example~\ref{ex:pref-concerts}.

\subsection{Results: computation of the output formula}\label{sec:exp-results}

In order to assess feasibility of the computation of the preference formula resulting after applying a sequence of operators, we report the corresponding execution time averaged out over 100 different runs (as measured on a machine sporting a 2,3 GHz 8-Core Intel Core i9 with 32 GB of RAM).

Our first experiments test the impact of the characteristics of the taxonomy in the case of synthetic taxonomies and one pair of conflicting preferences.
For regular taxonomies, computing $F^{\Top}$ ($0.5ms$ on average) is generally faster than computing $F^{\Sop\Top\Sop\Top}$ ($1.5ms$) and $F^{\Top\Sop\Top}$ ($2.7ms$) and neither $\fanout$ nor $\depth$ affect the computation time significantly.
Similar times are obtained with random taxonomies.
With scale-free taxonomies the same relative costs are kept, but times are slightly higher, due to the much deeper structure, and tend to decrease as the the number of maximal nodes increases, as shown in Figure~\ref{fig:powerLaw-formulaTimes}; still, all times 
are well under $0.2s$ and thus
negligible with respect to the time required for computation of $\Best$, as will be shown in Section~\ref{sec:best-experiments}.

\begin{figure}%
\centering
\subfloat[][{Scale-free taxonomies.}]
{\includegraphics[width=.44\columnwidth]{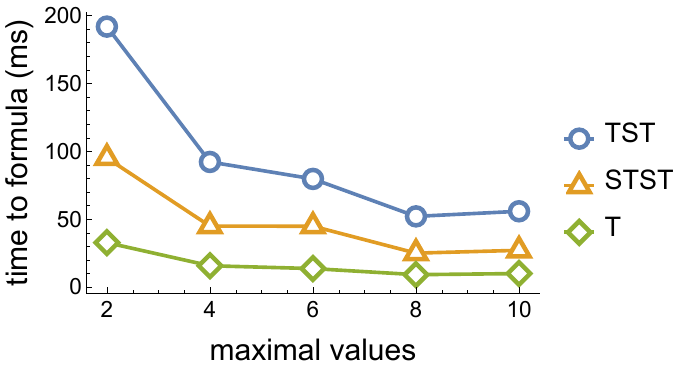}\label{fig:powerLaw-formulaTimes}}%
\quad
\subfloat[][{Regular taxonomies.}]
{\includegraphics[width=0.48\columnwidth]{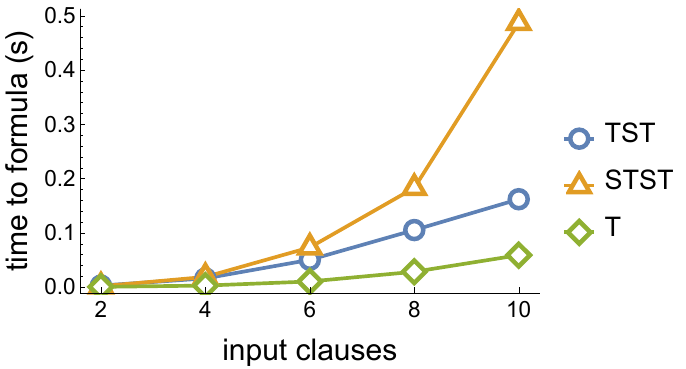}\label{fig:esperimenti_risultati_figure_VaryingPairs-times}}%
\\[-2ex]
\subfloat[][{Contextual preferences.}]
{\includegraphics[width=0.48\columnwidth]{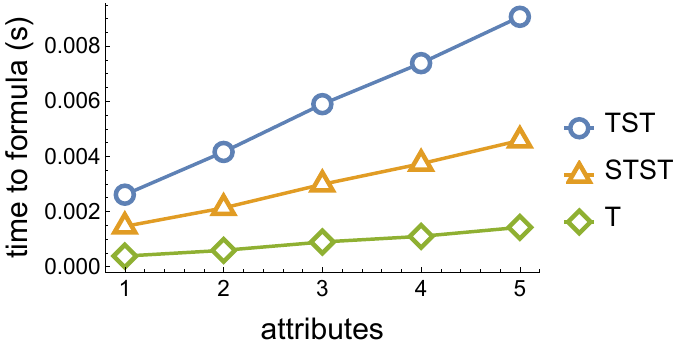}\label{fig:esperimenti_risultati_figure_VaryingContextualPrefs-times}}%
\quad
\subfloat[][{\flipkart taxonomy.}]
{\includegraphics[width=0.48\columnwidth]{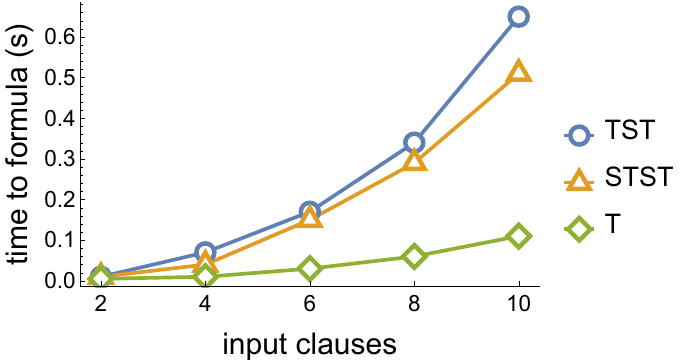}\label{fig:esperimenti_risultati_figure_FlipkartVaryingPairs-times}}%
\\[-2ex]
\caption{Time for computing the formula: various settings.}\label{fig:esperimenti_risultati_figure_VaryingPairs-formula}%
\end{figure}

Figure~\ref{fig:esperimenti_risultati_figure_VaryingPairs-times} shows that the time for computing the formula grows with the number of input clauses,
with times always below $0.5s$.

For a multi-attribute scenario, Figure~\ref{fig:esperimenti_risultati_figure_VaryingContextualPrefs-times} shows the behavior with contextual preferences as the number of attributes varies.
The resulting formula is always computed in less than $0.01s$; times slightly grow as the number of attributes grows, but remain low.

We now turn to the case of real taxonomies.
With \usedcars, which is functional, results are very similar to those obtained with synthetic taxonomies and thus not shown here in the interest of space.
We then test on \flipkart, which is non-functional, the case of conflicting preferences as the number of input clauses $\clauses$ varies.
This has an impact on the overhead for determining redundancies in formulas and for checking clause satisfiability when computing $\Top$. Indeed, both require checking whether two values $v_1$ and $v_2$ have a common descendant in the taxonomy, which is immediate in the case of functional taxonomies, as it suffices to check whether there is a path from $v_1$ to $v_2$ or vice versa.
However, for non-functional taxonomies this check may require extracting all descendants of $v_1$ and $v_2$, which may be expensive for large taxonomies, especially when $v_1$ and $v_2$ are maximal values.
Yet, for taxonomies in which only few nodes have more than one parent (like \flipkart, with 170 such nodes), it is convenient to keep track of those nodes at taxonomy load time; with this, we can check the existence of a common descendant between $v_1$ and $v_2$ by checking whether there is a path to both from one of those nodes (if they are not the descendant of one another).
As Figure~\ref{fig:esperimenti_risultati_figure_FlipkartVaryingPairs-times} shows, the times measured with the \flipkart taxonomy are only slightly higher than with synthetic taxonomies (and always sub-second).

\subsection{Results: computation of $\Best$}\label{sec:best-experiments}

\begin{figure}%
\centering
\subfloat[][{Cardinality of $\Best$.}]
{\includegraphics[width=0.45\columnwidth]{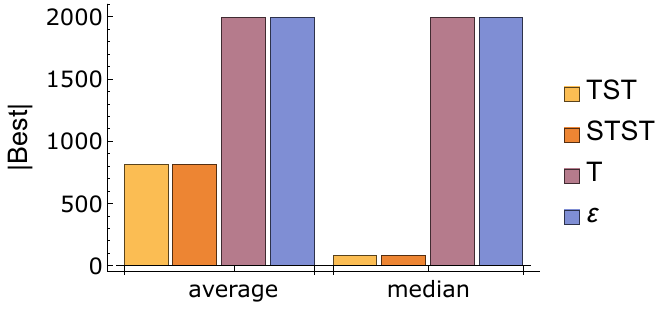}\label{fig:esperimenti_risultati_figure_Default-Best}}%
\quad
\subfloat[][{Time for computing $\Best$.}]
{\includegraphics[width=0.45\columnwidth]{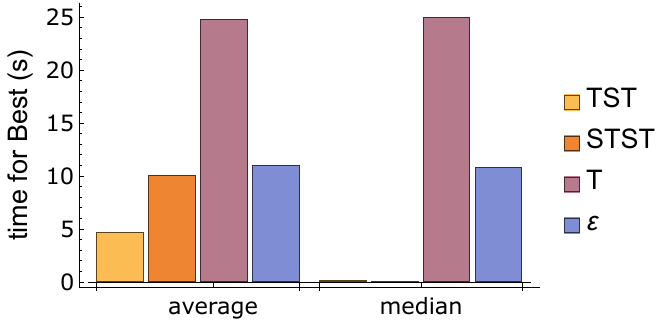}\label{fig:esperimenti_risultati_figure_Default-BestTimes}}%
\\[-2ex]
\caption{Computing $\Best$ with default parameter values.\label{fig:esperimenti_risultati_figure_Default}}%
\end{figure}

As discussed in Section~\ref{sec:algorithm}, we restrict the $\Best$ operator to act only on \emph{relevant} t-tuples.
In the same vein, we shall only consider preferences inducing a non-empty set of relevant t-tuples.

With conflicting preferences and default parameter values on regular taxonomies, the amount of relevant t-tuples
is roughly $40\%$ of the size of a synthetic dataset.
Figure~\ref{fig:esperimenti_risultati_figure_Default-Best} shows that
both 
 $\Top$ 
and $\Emptyseq$ retain about half of the relevant t-tuples (which is both the average and the median value we obtained),
while $\Top\Sop\Top$ and $\Sop\Top\Sop\Top$ retain less than $2\%$ in the median case (the average value goes up to $20\%$ due to runs with unfocused input formulas referring to values not in the dataset).
This is reflected in the computation times, shown in Figure~\ref{fig:esperimenti_risultati_figure_Default-BestTimes}, which are consistently around $24s$ for $\Top$ and $10s$ for $\Emptyseq$, but nearly two orders of magnitude smaller in the median case for $\Top\Sop\Top$ and $\Sop\Top\Sop\Top$.
With both scale-free and random taxonomies, the amount of relevant t-tuples varies much more (with an average still around $40\%$), but times are on average one order of magnitude smaller for $\Top\Sop\Top$ and $\Sop\Top\Sop\Top$ than for $\Top$, with results for the latter covering almost the entire dataset due to the lack of conflict resolution.

We observe that the application of $\Top$ alone corresponds to the work performed by preference evaluation methods that only aim at guaranteeing transitivity, e.g., \cite{Kie02,Ch03,DBLP:journals/tkde/GolfarelliRB11}, which are therefore outperformed by our approach.
The inability of $\Top$ 
to deal with conflicting preferences, thus generating many indifferent t-tuples, which in turn induce (very) large result sets, 
indeed
applies to all our scenarios.
Similar observations apply to $\Emptyseq$ (i.e., the empty sequence, corresponding to the input formula), which represents the action of works on preference evaluation using no rewriting whatsoever, such as \cite{DBLP:conf/sigmod/ChanJTTZ06,DBLP:conf/ijcai/LukasiewiczMS13}. 
Additionally, the results obtained via $\Emptyseq$ would be totally unreliable, due to lack of transitivity (see Example~\ref{ex:non-transitive-algorithm}).
We thus refrain from considering $\Top$ and $\Emptyseq$ from now on.

We now analyze the cost incurred by the computation of $\Best$ as we deviate from standard parameter values.
In the case of contextual preferences, adding context makes the $\Best$ set leaner and, thus, easier to compute, so that times are under $1s$ already with two attributes.
As usual, $\Sop\Top\Sop\Top$ is slightly quicker to compute, since it gives rise to a smaller formula (although its strict version coincides with that of $\Top\Sop\Top$, and thus their cardinalities coincide).

As already visible in Figure~\ref{fig:esperimenti_risultati_figure_Default}, random preference formulas may fail to represent a meaningful specification of preferences, thus leading to very large result sets.
For this reason, we disregard such formulas and, in particular,
in the next experiments we only retain those ``good runs'' in which either $\Top\Sop\Top$ or $\Sop\Top\Sop\Top$ produce less than 2\% of the t-tuples in the dataset.
Figure~\ref{fig:esperimenti_risultati_figure_HeuristicSortGOODRUNSVaryingPrefsBest} shows how the cardinality of $\Best$ varies, under these hypotheses and default parameter values, as the number of input clauses $\clauses$ varies, thus confirming that $\Sop\Top\Sop\Top$ typically leads to a smaller result than $\Top\Sop\Top$.

\begin{figure}%
\centering
\subfloat[][{Cardinality of $\Best$.}]
{\includegraphics[width=0.48\columnwidth]{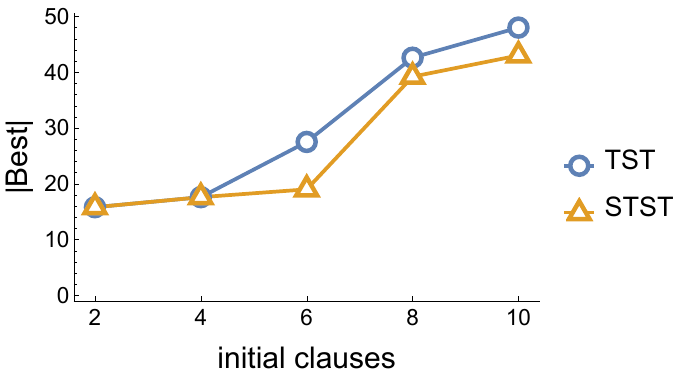}\label{fig:esperimenti_risultati_figure_HeuristicSortGOODRUNSVaryingPrefsBest}}%
\quad
\subfloat[][{Time for computing $\Best$ with and without heuristic sort (H).}]
{\includegraphics[width=0.48\columnwidth]{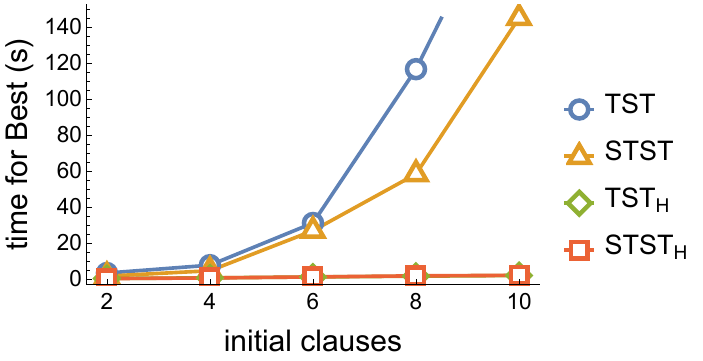}\label{fig:esperimenti_risultati_figure_HeuristicSortGOODRUNSVaryingPrefsTimes}}%
\\[-2ex]
\caption{Synthetic datasets: conflicting preferences, varying the number of input clauses $\clauses$ (only good runs).}\label{fig:esperimenti_risultati_figure_HeuristicSortGOODRUNSVaryingPrefs}%
\end{figure}

We now consider the heuristics described in Section~\ref{sec:algorithm}, which sorts the t-relation according to increasing height index values.
Figure~\ref{fig:esperimenti_risultati_figure_HeuristicSortGOODRUNSVaryingPrefsTimes} compares times obtained with the heuristic sort strategy (marked with an $H$ subscript) to those obtained with no heuristics as the number of input clauses $\clauses$ varies.
The sort takes between 3\% and 10\% of the total time spent for computing $\Best$, yet the use of the proposed heuristics largely outperforms standard executions, with times never exceeding $2s$; without the heuristics, times diverge to well over $100s$, on average, in the more expensive scenarios.

\begin{figure}%
\centering
\subfloat[][{Cardinality of $\Best$.}]
{\includegraphics[width=0.48\columnwidth]{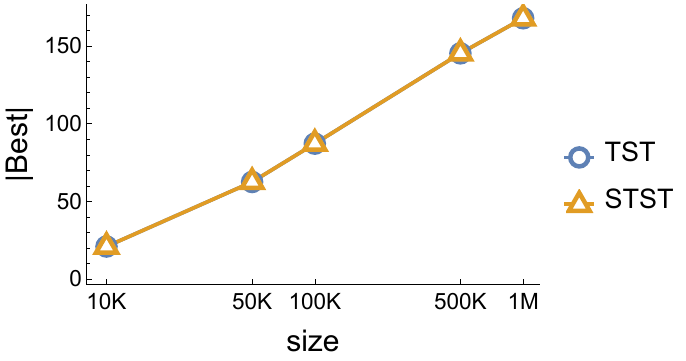}\label{fig:esperimenti_risultati_figure_VaryingSize1Pair-HS-Best}}%
\quad
\subfloat[][{Time for computing $\Best$ with heuristic sort (H).}]
{\includegraphics[width=0.48\columnwidth]{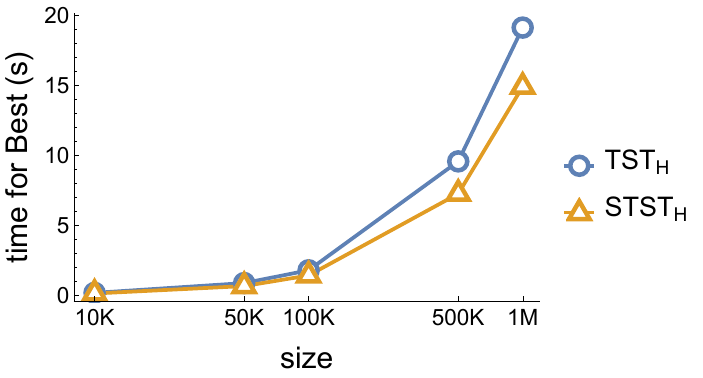}\label{fig:esperimenti_risultati_figure_VaryingSize1Pair-HS-BestTimes}}%
\\[-2ex]
\caption{Synthetic datasets: varying the dataset size $\size$ (only good runs, size in logarithmic scale).}\label{fig:esperimenti_risultati_figure_VaryingSize1Pair-HS}%
\end{figure}

Having ascertained the suitability of the heuristic sort, we demonstrate its scalability with the experiment shown in Figure~\ref{fig:esperimenti_risultati_figure_VaryingSize1Pair-HS}, which shows a linear trend for times as the size $\size$ of the dataset varies, while  cardinalities tend to grow logarithmically.

\begin{figure}%
\centering
\subfloat[][{Cardinality of $\Best$.}]
{\includegraphics[width=0.48\columnwidth]{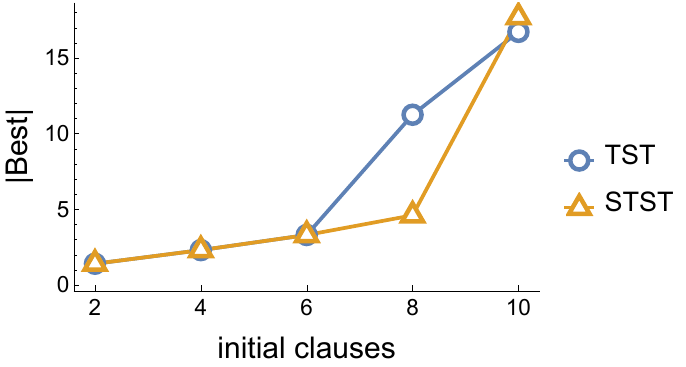}\label{fig:esperimenti_risultati_figure_FlipkartHeuristicSortGOODRUNSVaryingPairsBest}}%
\quad
\subfloat[][{Time for computing $\Best$ with heuristic sort (H).}]
{\includegraphics[width=0.48\columnwidth]{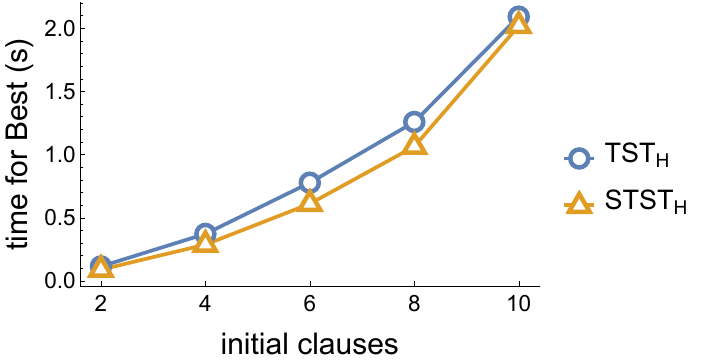}\label{fig:esperimenti_risultati_figure_FlipkartHeuristicSortGOODRUNSVaryingPairsTimes}}%
\\[-2ex]
\caption{\flipkart: conflicting preferences, varying the number of input clauses $\clauses$ (only good runs).}\label{fig:esperimenti_risultati_figure_FlipkartHeuristicSortGOODRUNSVaryingPairs}%
\end{figure}
\begin{figure}%
\centering
\subfloat[][{Cardinality of $\Best$.}]
{\includegraphics[width=0.48\columnwidth]{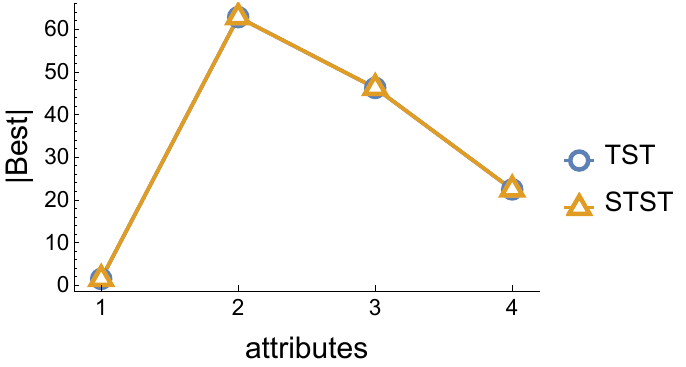}\label{fig:esperimenti_risultati_figure_FlipkartContextualHeuristicSortGOODRUNSVaryingAttributesBest}}%
\quad
\subfloat[][{Time for computing $\Best$ with heuristic sort (H).}]
{\includegraphics[width=0.48\columnwidth]{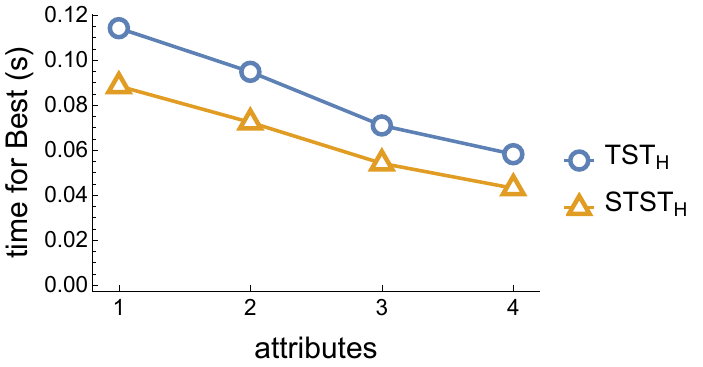}\label{fig:esperimenti_risultati_figure_FlipkartContextualHeuristicSortGOODRUNSVaryingAttributesTimes}}%
\\[-2ex]
\caption{\flipkart: contextual preferences, varying the number of attributes $\attributes$ (only good runs).}\label{fig:esperimenti_risultati_figure_FlipkartContextualHeuristicSortGOODRUNSVaryingAttributes}%
\end{figure}

The trends shown with synthetic data are confirmed with real data on \flipkart.
Figure~\ref{fig:esperimenti_risultati_figure_FlipkartHeuristicSortGOODRUNSVaryingPairsBest} shows that the cardinality of $\Best$ typically grows as the number of input clauses grows.
Consequently, Figure~\ref{fig:esperimenti_risultati_figure_FlipkartHeuristicSortGOODRUNSVaryingPairsTimes} shows times slightly growing with the number of input clauses, but always under $2.1s$.
The case of contextual preferences is shown in Figure~\ref{fig:esperimenti_risultati_figure_FlipkartContextualHeuristicSortGOODRUNSVaryingAttributes}, where times decrease as the number of attributes grows, since the number of relevant t-tuples decreases with the number of applied contexts.
For the same reason, the cardinality of $\Best$ is higher with 2 or 3 attributes than with 4; however, with only 1 attribute (and thus no context) the cardinality is the lowest, since the t-tuples satisfying the most specific preference are not filtered out by contexts. 

\begin{figure}
    \centering
    \begin{tabular}{cc}
    \adjustbox{valign=b}{\subfloat[User interface.\label{fig:interface}]{%
          \includegraphics[trim={0 0 0 5.5cm},clip,width=0.45\columnwidth]{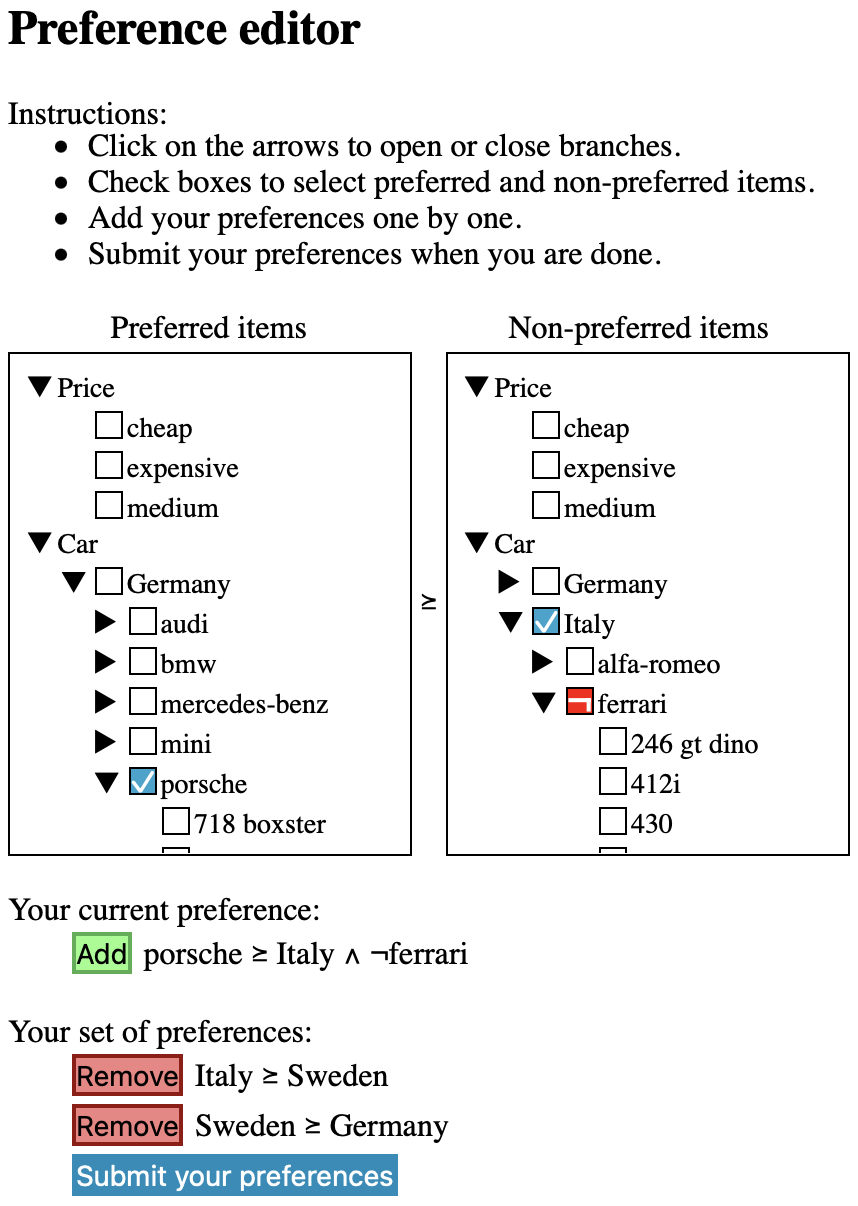}}}
    &
    \adjustbox{valign=b}{\begin{tabular}{@{}c@{}}
    \subfloat[Cardinality of $\Best$.\label{fig:SmallCarsCards}]{%
          \includegraphics[width=0.45\columnwidth]{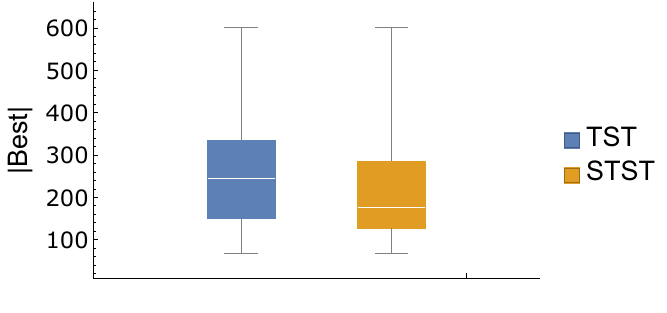}} \\[-2ex]
    \subfloat[Time for computing $\Best$.\label{fig:SmallCarsTimes}]{%
          \includegraphics[width=0.45\columnwidth]{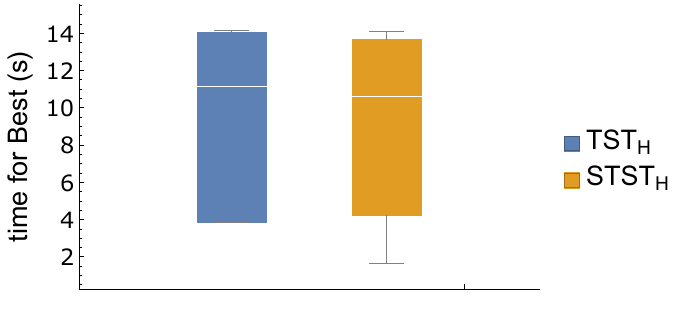}}
    \end{tabular}}
    \end{tabular}
    \caption{User interface and experiments on \usedcars.}\label{fig:usedcars}
\end{figure}

For experiments on \usedcars, we collected preferences from a set of 107 users by means of a Web interface allowing the specification of statements in the simplified notation presented in Section~\ref{sec:preliminaries} through an expandable tree-view of the taxonomy
 (see Figure~\ref{fig:interface}).
After instructing users on how to specify preferences (even conflicting ones), we observed an average of $3.4$ statements (from 2 to 9) per query and as many as 78\% of cases of conflicts.
Figure~\ref{fig:SmallCarsCards} shows box plots representing the distributions of cardinalities of $\Best$ obtained with user-defined preferences, which confirms that $\Sop\Top\Sop\Top$ tends to produce slightly smaller results than $\Top\Sop\Top$.
We observe that such cardinalities, typically corresponding to the number of cars available for a specific model and price range, are very low with respect to the dataset size (and could be further reduced if filters based on other criteria, such as mileage, were applied).
Execution times (Figure~\ref{fig:SmallCarsTimes}) are, on average, below $10s$ for both sequences, and thus
overall acceptable and comparable with the measurements obtained with similarly sized synthetic data (Figure~\ref{fig:esperimenti_risultati_figure_VaryingSize1Pair-HS-BestTimes}).

\section{Related Works and Discussion}\label{sec:related}

In spite of the many works on the use of \emph{qualitative} preferences for querying databases (see, e.g., \cite{DBLP:journals/tods/StefanidisKP11}), only a few address the issues arising when attributes' domains exhibit a hierarchical structure.

Preferences in OLAP systems are considered in \cite{DBLP:journals/tkde/GolfarelliRB11}, where an algebraic language, based on that in \cite{Kie02},
is adopted. Preferences on attributes are only of an \emph{absolute} type, stating which are the most (resp.\ least) preferred values at a given ``level'' of a dimensional attribute. Preferences are then propagated along levels, with no concern for the combination of preferences, less so conflicting ones.

Lukasiewicz et al. \cite{DBLP:conf/ijcai/LukasiewiczMS13}
extend the Datalog+/- ontological language with qualitative preferences, yet they do not address the problems arising from conflicting preferences. In a subsequent work \cite{DBLP:journals/jodsn/LukasiewiczMST15}, the authors assume that, besides the order generated by the preferences, another linear order exists, originating from probabilistic scores attached to specific objects. Since the two orders may conflict, ad-hoc operators for compromising among the two orders are introduced and evaluated. Although \cite{DBLP:journals/jodsn/LukasiewiczMST15} considers conflicts, these are \emph{not} among preferences and their solutions are not applicable to the scenario we consider in this paper.

To the best of our knowledge, no other work addresses the exact same issues we tackle here. Yet, Section~\ref{sec:best-experiments} has shown how existing methods (those that just enforce transitivity as well as works on preference evaluation using no rewritings) would be unsuitable to meet the goals we set in this paper.

The specificity principle on which we have based the definition of our \Sop\ operator follows a long-standing tradition in the AI and KR fields,
in which conflicts arising from contradictory evidences (antecedents) are solved by means of non-monotonic reasoning.
However, in this context, the issue of 
 inheritance of properties, which can be dealt with in different ways according to the adopted reasoning theory (see, e.g., \cite{Horty94}), leads to problems that are quite different from those we have considered in this paper.

The need to address conflicts arising from preferences was also observed in~\cite{DBLP:conf/er/CiacciaMT19}. The framework proposed there allows for a restricted form of taxonomies (with all values organized into distinct, named levels) and hints at an ad hoc procedure with very limited support for conflict resolution;  the focus of~\cite{DBLP:conf/er/CiacciaMT19} is, however, on the downward propagation of preferences.

A kind of specificity principle was also considered in \cite{DBLP:journals/jacm/CiacciaMT20}, albeit on a different preference model (using strict rather than weak preferences) and a different scenario, in which preferences are to be combined across different \emph{contexts}~\cite{DBLP:conf/fqas/MartinenghiT09}. In that work, given two conflicting preferences, e.g., $a \pref b$, which is valid in a context $c$, and $b \pref a$ valid in context $c'$, if context $c$ is more specific than $c'$ then $a \pref b$ wins and $b \pref a$ is discarded. Thus, specificity considered in \cite{DBLP:journals/jacm/CiacciaMT20} concerns contexts,
whereas, in the present paper, specificity has to do with preference statements that involve values at different levels of detail in the taxonomies.
Conflicts in~\cite{DBLP:journals/jacm/CiacciaMT20} are at the level of a single pair of objects (since no language for specifying preferences was considered there), whereas in the present work we deal with conflicts between \emph{preference statements}, which in general involve many pairs of objects - a fact that requires a solution incomparable with those adopted in~\cite{DBLP:journals/jacm/CiacciaMT20}.

A line of research that is only apparently related to ours concerns the problem of propagating preferences across the nodes/terms of an ontology, see, e.g., 
\cite{DBLP:journals/isci/CenaLO13,DBLP:conf/ausai/ChamielP08,DBLP:journals/jiis/Martinez-Garcia19}. Given ``interest scores'' attached to some terms, these works focus on (numerical) methods to combine and propagate such scores to ``similar'' terms in the ontology.

A definitely relevant issue, orthogonal to our focus and thus outside the scope of this paper, is that of \emph{preference elicitation}. This problem has been thoroughly studied in various fields, such as Recommender Systems, decision making, marketing, and behavioral economics, with remarkable recent attention on \emph{relative preferences}, either expressed with pairwise comparisons or inferred from absolute preferences~\cite{DBLP:conf/recsys/Kalloori0G18,DBLP:conf/sigir/KallooriL019}.

Common methods to solve conflicts among preferences are based on the use of
 operators, the most well-known being Pareto and Prioritized composition \cite{Ch03,Kie02,DBLP:journals/jacm/CiacciaMT20}. Given a conflict between $a$ and $b$ originating from two different preference statements, Pareto composition just drops both preferences $a \pref b$ and $b \pref a$. Conversely, Prioritized composition \emph{a priori} assumes that one of the two statements is more important than the other, and then solves the conflict by retaining the corresponding preference. We have no such a-priori notion of priority, which might be hard to define in practice; rather, we rely on a definition of specificity that \emph{dynamically} determines if a statement takes precedence over another depending on the available taxonomies.

Many algorithms have been devised to answer preference queries, although most of them work only for numerical attributes \cite{KoIo04}. Among the algorithms that can be applied to arbitrary strict partial orders $\pref$, BNL \cite{BKS01} is undoubtedly the most well-known among those that compute the result sets by means of dominance tests. Improvements to the BNL logic, such as those found in the SFS \cite{DBLP:conf/icde/ChomickiGGL03} and SaLSa \cite{BaCP08} algorithms, require the input relation to be topologically sorted, which in these algorithms is based on the presence of numerical attributes.
A different approach, pioneered in \cite{DBLP:conf/icde/GeorgiadisKCNS08}, avoids (most of the) dominance tests by partitioning the domain of (relevant) tuples into a set of equivalence classes, where each class includes all and only those tuples whose values are the best for a subset of the input preference statements. For instance, a statement like 
$(A_i = v) \wpref (A_i = v')$ induces two equivalence classes, the first including all tuples with value $v$ for attribute $A_i$, and the second those with value $v'$. For each equivalence class a different SQL query is then executed, until it is guaranteed that no further optimal tuples exist. However, since the number of equivalence classes is exponential in the number of input statements, this approach cannot be adopted in our framework, in which the rewritten formula to be evaluated, due to the transitive closure operator, can well contain tens of statements.

Common practice typically focuses on the specification of \emph{quantitative} preferences, for instance by means of a function expressing a score based on the attribute values, as is commonly done in top-$k$ queries~\cite{DBLP:journals/csur/IlyasBS08,DBLP:journals/pvldb/MartinenghiT10,DBLP:journals/tkde/MartinenghiT12}. Recent works have tried to combine the qualitative nature of (Pareto) dominance with the quantitative aspects of ranking~\cite{DBLP:journals/tods/CiacciaM20,DBLP:journals/pvldb/CiacciaM17,DBLP:conf/cikm/CiacciaM18,DBLP:conf/sebd/CiacciaM18,DBLP:conf/sisap/BedoCMO19,DBLP:conf/sebd/CiacciaM19,DBLP:conf/sigmod/MouratidisL021,CM:PACMMOD2024}.

The specification of preferences is sometimes expressed through constraints of a ``soft'' nature, i.e., which can be violated. It should be interesting to combine the effects of the specification of ``hard'' constraints, such as the integrity constraints of a database, commonly adopted for query optimization and integrity maintenance~\cite{M:PHD2005,DBLP:conf/lopstr/ChristiansenM03,DBLP:conf/foiks/ChristiansenM04,DBLP:conf/fqas/Martinenghi04,DBLP:conf/dexa/MartinenghiC05,DBLP:journals/aai/ChristiansenM00,DBLP:conf/dexaw/DeckerM07,DBLP:conf/lpar/DeckerM06,DBLP:conf/dexaw/DeckerM06,DBLP:journals/tplp/CaliM10}, or even of structural constraints governing access to data~\cite{CM:ICDE2008,DBLP:journals/tplp/CaliM10,DBLP:conf/er/CaliM08,CCM:JUCS2009}, with the techniques studied in this paper for retrieving the best options.
A more tolerant approach consists in coping with the presence of inconsistent or missing values~\cite{DBLP:conf/lpar/DeckerM06,DBLP:conf/dexaw/DeckerM06,DBLP:conf/dexaw/DeckerM07,DBLP:conf/ppdp/DeckerM08}; in such cases, it would be interesting to understand how the amount of such an inconsistency in the data~\cite{DBLP:journals/jiis/GrantH06,DBLP:conf/er/DeckerM09,DBLP:conf/ijcai/GrantH11,DBLP:conf/ecsqaru/GrantH13,DBLP:journals/ijar/GrantH17,DBLP:journals/ijar/GrantH23} may affect the results.

We also observe that preference elicitation and management are typical parts of data preparation pipelines, which might then involve data subsequently processed by Machine Learning algorithms~\cite{DBLP:conf/fqas/Masciari09,DBLP:journals/isci/MasciariMZ14,DBLP:conf/ismis/MasciariMPS20,DBLP:conf/ideas/MasciariMPS20} and retrieved from heterogeneous sources, including  RFID~\cite{DBLP:conf/ideas/FazzingaFMF09,DBLP:journals/tods/FazzingaFFM13}, pattern mining~\cite{DBLP:conf/ideas/MasciariGZ13}, crowdsourcing applications~\cite{DBLP:conf/socialcom/GalliFMTN12,DBLP:conf/www/BozzonCCFMT12,DBLP:conf/mmsys/LoniMGGMAMMVL13}, and streaming data~\cite{DBLP:journals/jiis/CostaMM14}.

\section{Conclusions}\label{sec:conclusions}

In this paper we have tackled the problem of finding the best elements from a repository on the basis of preferences referring to values that are more generic than the underlying data and may involve conflicts. To this aim, we have introduced and formally investigated two operators for enforcing, in a given collection of preferences, the properties of specificity, which can solve conflicts, and transitivity, which guarantees the soundness of the final result. We have then characterized the limitations that can arise from their combination and identified the best ways in which they can be used together. We have finally proposed a technique based on an original heuristics for selecting the best results associated with given sequences of operators and shown, with a number of experiments over both synthetic and real-world datasets, the effectiveness and practical feasibility of the overall approach.
Future work includes extending our framework to more general scenarios in which domain values are connected by ontological relationships, as is the case in Ontology-Based Data Access~\cite{DBLP:conf/ijcai/XiaoCKLPRZ18}.

\bibliographystyle{plain}

\appendix
\clearpage
\section{Proofs}

\thmBestXS*
\begin{proof}
By contradiction, assume $\exists\ r$ and a t-tuple $t_1 \in r$ with $t_1 \in \Best_{\pref_{\Xop\Sop}}(r)$ and $t_1 \not\in \Best_{\pref_{\Xop}}(r)$. 
This implies that $r$ includes a t-tuple $t_2$ such that $t_2 \pref_{\Xop} t_1$, yet $t_2 \not\pref_{\Xop\Sop} t_1$. Since $\wpref_{\Xop\Sop} \subseteq\wpref_{\Xop}$, we have that the preference $(t_2,t_1)$ has been removed from $\wpref_{\Xop}$ by the \Sop\ operator. But this in turn implies that $t_1 \wpref_{\Xop} t_2$ (otherwise there would have been no conflict), thus $t_1 \equp_{\Xop} t_2$, which contradicts the hypothesis that $t_2 \pref_{\Xop} t_1$.
\end{proof}

\thmBasicProperties*
\begin{proof}
\emph{Idempotence}. Idempotence of transitive closure ($\Top$) is well-known.
As for $\Sop$,
by its construction and thanks to Lemma~\ref{lem:implication},
after the rewriting caused by $\Sop$, none of the resulting clauses is more specific than any other clause; therefore, another application of $\Sop$ would be idle.

\emph{Inflation/deflation} and \emph{monotonicity}. Directly from the definitions of the operators.

\emph{Maximality} of $\Top$.
We prove it by induction on the length of the sequence.
\emph{Base case}.
For the only sequence of length 0 ($\Emptyseq$), we have $\Emptyseq \sqsubseteq \Top$ by  inflation on $\Emptyseq$.
\emph{Inductive step}. By inductive hypothesis, assume that $\Yop \sqsubseteq \Top$, where $|\Yop|=n$. We show that containment also holds for all sequences of length $n+1$, i.e., $\Yop\Top$ and $\Yop\Sop$.
By monotonicity of $\Top$, we have $\Yop\Top \sqsubseteq \Top\Top$ and, by idempotence, we obtain $\Yop\Top \sqsubseteq \Top$.
Analogously, $\Yop\Sop \sqsubseteq \Top$ thanks to the inductive hypothesis and deflation.
\end{proof}

\thmXTcontainsXTST*
\begin{proof}
By deflation, we have $\Xop\Top\Sop \sqsubseteq \Xop\Top$.
By monotonicity, $\Xop\Top\Sop\Top \sqsubseteq \Xop\Top\Top$.
The thesis follows from idempotence of $\Top$.
\end{proof}

\thmXTSTstrictContainsXTstrict*
\begin{proof}
Let $(a,b)$ be a preference that holds in $\pref_{\Xop\Top}$.
This means that $(a,b)$ holds in $\wpref_{\Xop\Top}$ while $(b,a)$ does not hold in $\wpref_{\Xop\Top}$.
Since there is no preference opposite to (and more specific than) $(a,b)$ in $\wpref_{\Xop\Top}$, the application of $\Sop$ cannot remove $(a,b)$, which then also holds in $\wpref_{\Xop\Top\Sop}$, while $(b,a)$ continues not to hold in $\wpref_{\Xop\Top\Sop}$, since $\Sop$ cannot add preferences.
So, $(a,b)$ also holds in $\wpref_{\Xop\Top\Sop\Top}$, since $\Top$ does not remove any preferences.
Moreover, $(b,a)$ cannot hold in $\wpref_{\Xop\Top\Sop\Top}$, since it does not hold in $\wpref_{\Xop\Top}$ and $\Xop\Top\Sop\Top\sqsubseteq\Xop\Top$ by Lemma~\ref{thm:XT-contains-XTST}.
Therefore, $(a,b)$ holds in $\pref_{\Xop\Top\Sop\Top}$.
\end{proof}

\thmTSrepeated*
\begin{proof}
It suffices to show that $\Top\Sop\equiv\Top\Sop\Top\Sop$, since the relationship must hold for any initial formula $F$.
If, for any $F$, $F^\Top=F^{\Top\Sop}$ or $F^\Top=F^{\Top\Sop\Top}$ or $F^{\Top\Sop}=F^{\Top\Sop\Top}$ then $\Top\Sop\equiv\Top\Sop\Top\Sop$ by idempotence.
So assume $F^\Top\neq F^{\Top\Sop}$, $F^\Top\neq F^{\Top\Sop\Top}$ and $F^{\Top\Sop}\neq F^{\Top\Sop\Top}$; so $\wpref_{\Top\Sop}$ is not transitive.
So there must be three t-tuples $a$, $b$, $c$ that violate transitivity, i.e.,
$a\wpref_{\Top\Sop} b$ and $b \wpref_{\Top\Sop} c$, but $a\not\wpref_{\Top\Sop} c$.
Since $\wpref_{\Top}$ is transitive and $\Sop$ cannot add any preference, $\wpref_{\Top}$ must also contain the preferences $(a,b)$, $(b,c)$ and $(a,c)$.
The removal of $(a,c)$ means that in $\wpref_{\Top}$ it was in conflict with $(c,a)$, and less specific. So $\wpref_{\Top}$ must also contain $(c,a)$ and, by transitivity, $(c,b)$ and $(b,a)$.
Clearly, $(a,c)$ must be in $\wpref_{\Top\Sop\Top}$, since $\wpref_{\Top\Sop}$ contains $(a,b)$ and $(b,c)$.
So, everything that was deleted from $\wpref_{\Top}$ by violating transitivity of $\wpref_{\Top\Sop}$ is restored in $\wpref_{\Top\Sop\Top}$.

We prove the theorem in two parts.
Part 1: what was added back to $\wpref_{\Top\Sop\Top}$ in order to restore transitivity (i.e., $(a,c)$) is removed from $\wpref_{\Top\Sop\Top\Sop}$.
Part 2: any pair $(x,y)$ that is removed from $\wpref_{\Top}$ in $\wpref_{\Top\Sop}$ and is not the outcome of transitivity is not added back to $\wpref_{\Top\Sop\Top}$.
So $\wpref_{\Top\Sop}$ and $\wpref_{\Top\Sop\Top\Sop}$ coincide.

Let us write $P^R$ to denote the same as statement $P$, but with swapped arguments; let us write $P\subset P'$ to indicate that $P$ is subsumed by $P'$ and, consequently, $P^R\subset P'$ to indicate that $P$ is more specific with respect to $P'$. Let us also call the \emph{left set} of $P$, denoted $LS(P)$, the set of tuples $\{x \mid P(x,y)\}$ and the \emph{right set} of $P$, denoted $RS(P)$, the set $\{y \mid P(x,y)\}$.
Let $F$ be the initial formula corresponding to $\wpref$.

\textbf{Part 1}: $(a,c)$ is in $\wpref_{\Top}$, not in $\wpref_{\Top\Sop}$, and in $\wpref_{\Top\Sop\Top}$. We show that it is not in $\wpref_{\Top\Sop\Top\Sop}$.
Let $(a,c)$ be generated by statement $P_{ac}$, which is transitively obtained via two statements $P_{ab}$ and $P_{bc}$.
Let $(c,a)$ be generated by at least one statement more specific than $P_{ac}$ -- call it $P_{ca}$, so $P_{ca}^R\subset P_{ac}$.

In $F^{\Top\Sop}$, via specificity, $P_{ac}$ is replaced by a statement $P'_{ac}$ such that $(a,c)$ does not hold, and thus $P'_{ac}\subset P_{ac}$.
In $F^{\Top\Sop}$, we also have $P_{ca}$, since, if $P_{ca}$ was replaced by another statement $P'_{ca}$, then in $F^{\Top}$ we should have a statement $P''_{ac}\subset P_{ca}$; but then, since $P^R_{ca}\subset P_{ac}$, we would have $P''_{ac}\subset P_{ac}$, so $P''_{ac}$ would be subsumed and, thus, removed.
If in $F^{\Top\Sop}$ we have $P_{ab}$ and $P_{bc}$, then in $F^{\Top\Sop\Top}$ we have $P_{ac}$ again, so specificity between $P_{ca}$ and $P_{ac}$ will remove $(a,c)$ again from $\wpref_{\Top\Sop\Top\Sop}$.
If not, then at least one of $P_{ab}$ and $P_{bc}$ was replaced by a more specific statement in $F^{\Top\Sop}$, call them $P'_{ab}$ and $P'_{bc}$.
If $P'_{ab}$ does not coincide with $P_{ab}$, that means that in $F^{\Top}$ there is at least a statement $P''^R_{ba}\subset P_{ab}$ and $P'_{ab}$ is obtained via specificity of $P''_{ba}$ wrt $P_{ab}$; however, $(b,a)$ cannot hold in $P''_{ba}$, or else $(a,b)$ would not hold in $P'_{ab}$, against our hypotheses that $(a,b)$ holds in $\wpref_{\Top\Sop}$. Similarly for $P'_{bc}$, with $P''^R_{cb}\subset P_{bc}$.

In $F^{\Top\Sop\Top}$, there must be a statement $P^*_{ac}$ transitively obtained through $P'_{ab}$ and $P'_{bc}$, so that $(a,c)$ holds in $\wpref_{\Top\Sop\Top}$.
Statement $P_{ca}$ is preserved, since it is not replaced in $F^{\Top\Sop}$, as shown, and it cannot be subsumed in  $F^{\Top\Sop\Top}$, or it would also have been subsumed in $F^{\Top}$.
The only way in which $P_{ca}$ could not be more specific than $P^*_{ac}$ is to have $P''_{ba}$ or $P''_{cb}$ in $F^{\Top}$.
If $P''_{ba}$ removed (by specificity) from the right set of $P_{ab}$ all its intersection with the left set of $P'_{bc}$ then $(a,c)$ would not hold, against hypotheses. So, it cannot remove all of it. Similarly, it cannot remove all the left set of $P_{ab}$, or $(a,c)$ would not hold, against hypotheses. But then, $P'_{ab}$ is formed by two sets of pairs:
\[
\begin{array}{ll}
	LS(P_{ab}) \wpref RS'\\
	LS' \wpref RS(P_{ab}),
\end{array}
\]
where $LS'$ is the same as $LS(P_{ab})$ without the part removed by $P''_{ba}$ (and similarly for $RS'$).
Similarly, $P'_{bc}$ is formed by:
\[
\begin{array}{ll}
	LS(P_{bc}) \wpref RS''\\
	LS'' \wpref RS(P_{bc}),
\end{array}
\]
where $LS''$ is the same as $LS(P_{bc})$ without the part removed by $P''_{bc}$ (and similarly for $RS''$).
By combining the first part of $P'_{ab}$ with the second of $P'_{bc}$ (which can necessarily be combined, or else $(a,c)$ would not be obtained), we get
\[
\begin{array}{ll}
	LS(P_{ab}) \wpref RS(P_{bc}),
\end{array}
\]
which is $P_{ac}$, so $P^*_{ac}$ is indeed $P_{ac}$ and $P^R_{ca}\subset P_{ac}$. So, again, by specificity, $P_{ac}$ is going to be replaced by $P'_{ac}$, in which $(a,c)$ does not hold.

\textbf{Part 2:}
Let $(x,y)$ be a pair in $\wpref_{\Top}$ but not in $\wpref_{\Top\Sop}$ and that is not the outcome of a transitivity triple.
Then $(y,x)$ is also in $\wpref_{\Top}$, derived through a statement $P_{yx}$ such that $P^R_{yx}\subset P_{xy}$, both in $F^{\Top}$ and $F^{\Top\Sop}$.
By our assumptions, there are no two statements $P_1$ and $P_2$ in $F^\Top$ that can be transitively combined such that $(x,y)$ holds in their combination $P_3$.
Then these statements cannot exist in $F^{\Top\Sop}$ either, since $\Sop$ cannot add new preferences.
Then these statements cannot exist in $F^{\Top\Sop\Top}$ either, since $(x,y)$ would be obtained by transitivity already in $\wpref_{\Top}$, against hypotheses.
So they cannot exist in $F^{\Top\Sop\Top\Sop}$ either, so $(x,y)$ is not in $\wpref_{\Top\Sop\Top\Sop}$.
\end{proof}

\thmFiniteRepresentatives*
\begin{proof}
The theorem follows immediately from Lemma~\ref{thm:TS-repeated}, which allows the elimination of repeated $\Top\Sop$ sub-sequences, and from idempotence of $\Top$ and $\Sop$, which allows removing consecutive repetitions of the same operator from a sequence.
\end{proof}

\thmTnonminimal*
\begin{proof}
We already proved the claim for $\Sop\Top$ in Example~\ref{ex:ST-non-minimal}.
By Theorem~\ref{thm:finite-representatives}, we only need to prove it for $\Top$, $\Top\Sop\Top$, and $\Sop\Top\Sop\Top$.

For $\Top$, it suffices to consider any non-transitive formula: since it becomes transitive after applying $\Top$, it must contain extra preferences with respect to $F$.

For $\Top\Sop\Top$, consider a formula $F$ consisting of $P_1=\jun \wpref \may$ and $P_2=\lnot\jun \wpref \lnot\may$. Then, $F^{\Top}$ also includes $P_3=\jun \wpref \lnot\may$ and $P_4=\lnot\jun \wpref \may$. In $F^{\Top\Sop}$, $P_2$ (less specific than $P_1$) is replaced by $P_5=\{\lnot\jun \wpref \lnot\jun \land \lnot\may, \lnot\jun \land \lnot\may \wpref \lnot\may\}$, so that $\may\wpref\jun$ does not hold. Finally, in $F^{\Top\Sop\Top}$, $P_2$ is restored by transitively combining $P_5$ with itself, so that $\may\wpref\jun$ holds and $\Top\Sop\Top$ is not minimal.
Since, in this case, $F^{\Sop\Top}=F^{\Top}$, then $\Sop\Top\Sop\Top$ is also not minimal.
\end{proof}

\thmSnontransitive*
\begin{proof}
By Theorem~\ref{thm:finite-representatives}, we only need to show the claim for $\Sop$, $\Top\Sop$, and $\Sop\Top\Sop$.

For $\Sop$, any formula producing a non-transitive preference relation with no conflicting preferences suffices to prove the claim.

For $\Top\Sop$, consider a formula $F$ consisting of
$P_1 = \apr \wpref \may$,
$P_2 = \juntwentyfour \wpref \lnot\aprten \land \lnot\jun$, and
$P_3 = \lnot\apr \land \lnot\jun \wpref \juntwentyfour$.
Then, $F^{\Top}$ consists of $P_2$, $P_3$ and the following 4 preference statements:
\!\!\!\!\[
\begin{array}{rrcllrrcll}
P_4= &\apr& \!\!\wpref\!\!& \juntwentyfour& \mbox{ ($P_1+P_3$}),\\
P_6= &\lnot\apr\land\lnot\jun& \!\!\wpref\!\!& \lnot\aprten\land\lnot\jun& \mbox{ ($P_3+P_2$}),\\
P_5=&\juntwentyfour& \!\!\wpref\!\!& \juntwentyfour& \mbox{ ($P_2+P_3$}),\\
P_7=&\apr& \!\!\wpref\!\!& \lnot\aprten\land\lnot\jun& \mbox{ ($P_1+P_6$}),
\end{array}
\]
while $P_1$ is removed, as it is subsumed by $P_7$.
In $F^{\Top\Sop}$, $P_2$ (less specific than $P_3$) is replaced by $P_8=\juntwentyfour\wpref\apr\land\lnot\aprten$ and $P_4$ (less specific than $P_8$) is replaced by $P_9=\aprten\wpref \juntwentyfour$.
Now, in $F^{\Top\Sop}$, the preference, say, $\mayseven \pref_{\Top\Sop} \juntwentyfour$ holds strictly, since $\mayseven \wpref_{\Top\Sop} \juntwentyfour$ holds in $P_3$ while $\juntwentyfour \wpref_{\Top\Sop} \mayseven$ does not hold.
Similarly, $\juntwentyfour \pref_{\Top\Sop} \aprfifteen$ holds, since $\juntwentyfour \wpref_{\Top\Sop} \aprfifteen$ holds in $P_8$ while $\aprfifteen \wpref_{\Top\Sop} \juntwentyfour$ does not hold. However, $\mayseven \pref_{\Top\Sop} \aprfifteen$ does not hold strictly, since both $\mayseven \wpref_{\Top\Sop} \aprfifteen$ (via $P_6$) and $\aprfifteen \pref_{\Top\Sop} \mayseven$ (via $P_7$) hold. Thus, $\Top\Sop$ is not transitive.

For $\Sop\Top\Sop$, consider a formula $F$ consisting of
$P_1 = \aprten \wpref \jun10$,
$P_2 = \lnot\juntwentyfour \land \lnot\aprten \wpref \jun$, and
$P_3 = \jun \wpref \apr$,
Then $F^{\Sop\Top\Sop}$ consists of:
$P_4=\lnot \apr \land \lnot \juntwentyfour \wpref \jun$,
$P_5= \aprten \wpref \jun$,
$P_6= \lnot \aprten \land \lnot \juntwentyfour \wpref \apr$,
$P_7= \jun \wpref \jun$,
$P_8= \aprten \wpref \apr$,
$P_9= \jun \wpref \apr \wpref \lnot \aprten$.
With this, 
both
$\aprten \pref_{\Sop\Top\Sop} \juntwentyfour$
abd $\juntwentyfour \pref_{\Sop\Top\Sop} \aprfifteen$
hold
but 
$\aprten \pref_{\Sop\Top\Sop} \aprfifteen$
does not hold, so $\Sop\Top\Sop$ is not transitive.
\end{proof}

\thmincomparable*
\begin{proof}
Example~\ref{ex:TS<>STS} showed incomparability of $\Sop\Top\Sop$ and $\Top\Sop$, i.e.,
$\Sop\Top\Sop\not\sqsubseteq\Top\Sop$,
$\Top\Sop\not\sqsubseteq\Sop\Top\Sop$.
By exploiting the containments of Figure~\ref{fig:contenimentiSequenze}, we immediately obtain the following results:
$\Sop\Top\Sop\Top\not\sqsubseteq\Top\Sop$,
$\Sop\Top\not\sqsubseteq\Top\Sop$,
$\Top\Sop\Top\not\sqsubseteq\Sop\Top\Sop$.
The same formula used in Example~\ref{ex:TS<>STS} also works to show
$\Top\Sop\not\sqsubseteq\Sop\Top$,
hence
$\Top\Sop\Top\not\sqsubseteq\Sop\Top$.

Consider now a formula $F$ consisting of
$P_1 = \summer \wpref \may$,
$P_2 = \spring \wpref \may$, and
$P_3 = \auno \wpref \spring$.
Then $F^{\Sop\Top\Sop}$ consists of $P_1$, $P_2$, $P_3$,
$P_4 = \spring \wpref \lnot \may \land \lnot \auno$, and
$P_5 = \summer \wpref \lnot \may \land \lnot \auno$.
Instead, $F^{\Top\Sop\Top}$ consists of 
$P_6 = \spring \wpref \lnot \may \land \summer$,
$P_7 = \summer \wpref \lnot \may$,
$P_8 = \spring \wpref \spring$, and
$P_9 = \summer \wpref \spring$.
Therefore,
$\aprten \wpref_{\Sop\Top\Sop} \summer \land \lnot \auno$, while
$\aprten \not\wpref_{\Top\Sop\Top} \summer \land \lnot \auno$, and hence
$\Sop\Top\Sop\not\sqsubseteq\Top\Sop\Top$.
From this, we can immediately conclude that
$\Sop\Top\Sop\Top\not\sqsubseteq\Top\Sop\Top$,
$\Sop\Top\not\sqsubseteq\Top\Sop\Top$,
$\Sop\Top\Sop\Top\not\sqsubseteq\Top\Sop$,
$\Sop\Top\not\sqsubseteq\Top\Sop$,
which completely proves the claim.
\end{proof}

\thmMCT*
\begin{proof}
	By Theorem~\ref{thm:finite-representatives} and Lemma~\ref{thm:S-non-transitive}, the only complete and transitive sequences are $\Sop\Top$, $\Top\Sop\Top$, and $\Sop\Top\Sop\Top$. Among these, $\Sop\Top$ is not minimal, since $\Sop\Top\Sop\Top\sqsubseteq\Sop\Top$, while $\Top\Sop\Top$ and $\Sop\Top\Sop\Top$ are incomparable by Theorem~\ref{thm:incomparable}, hence the claim.
\end{proof}

\thmdiffBest*
\begin{proof}
The example in the proof of Theorem \ref{thm:incomparable} can be used to define a t-relation $r$ with $n$ t-tuples for which $\Best_{\Top\Sop\Top}(r) = r$ whereas $\Best_{\Sop\Top\Sop\Top}(r)$ consists of a single t-tuple with value \val{apr10}.

For the opposite case, the scenario in Example~\ref{ex:TS<>STS} can be used to define a t-relation $r$ with $n$ t-tuples for which $\Best_{\Sop\Top\Sop\Top}(r) = r$ whereas $\Best_{\Top\Sop\Top}(r)$ consists of a single t-tuple with value \val{jul21}.
\end{proof}

\end{document}